\documentclass[12pt, draftclsnofoot, onecolumn]{IEEEtran}
\normalsize

\usepackage{enumerate}
\usepackage{color}

\usepackage{cite}

\usepackage{float}
\newfloat{figtab}{htb}{fgtb}
\makeatletter
\newcommand\figcaption{\def\@captype{figure}\caption}
\newcommand\tabcaption{\def\@captype{table}\caption}
\makeatother

\usepackage{graphicx}
\ifCLASSINFOpdf
\else
\fi

\usepackage{amsthm}
\usepackage{amssymb,amsfonts}
\usepackage{amsfonts,balance}

\usepackage[cmex10]{amsmath}

\newtheorem{theorem}{Theorem}

\newtheorem{lemma}{\textbf{Lemma}}
\usepackage{algorithmic}

\usepackage{array}
\usepackage{multirow,makecell}
\usepackage[caption=false,font=footnotesize]{subfig}

\usepackage{verbatim}

\begin{document}
%
\title{Chemical Reactions-Based Microfluidic Transmitter and Receiver for Molecular Communication}
%
%
%

\author{Dadi~Bi,~\IEEEmembership{Student Member,~IEEE,}
        Yansha~Deng,~\IEEEmembership{Member,~IEEE,}\\
        Massimiliano~Pierobon,~\IEEEmembership{Member,~IEEE,} 
        and~ Arumugam~Nallanathan,~\IEEEmembership{Fellow,~IEEE}
\thanks{D. Bi and Y. Deng are with the Department of Informatics, King’s College London, London, WC2R 2LS, U.K. (e-mail:\{dadi.bi, yansha.deng\}@kcl.ac.uk). (Corresponding author: Yansha Deng).}
\thanks{M. Pierobon is with the Department of Computer Science and Engineering, University of Nebraska–Lincoln, Lincoln, NE 68588, USA (e-mail: pierobon@cse.unl.edu).}
\thanks{A. Nallanathan is with the School of Electronic Engineering and Computer
	Science, Queen Mary University of London, London, E1 4NS, U.K. (e-mail:
	a.nallanathan@qmul.ac.uk).}}

\maketitle

\vspace{-26pt}
\begin{abstract}
The design of communication systems capable of processing and exchanging information through molecules and chemical processes is a rapidly growing interdisciplinary field, which holds the promise to revolutionize how we realize computing and communication devices. While molecular communication (MC) theory has had major developments in recent years, more practical aspects in designing components capable of MC functionalities remain less explored. Motivated by this, we design a microfluidic MC system with a microfluidic MC transmitter and a microfluidic MC receiver based on chemical reactions. Considering existing MC literature on information transmission via molecular pulse modulation, the proposed microfluidic MC transmitter is capable of generating continuously predefined pulse-shaped molecular concentrations upon rectangular triggering signals using chemical reactions inspired by how cells generate pulse-shaped molecular signals in biology. We further design a microfluidic MC receiver capable of demodulating a received signal to a rectangular output signal using a thresholding reaction and an amplifying reaction. Our chemical reactions-based microfluidic molecular communication system is reproducible and well-designed, and more importantly, it overcomes the slow-speed, unreliability, and non-scalability of biological processes in cells. To reveal design insights, we also derive the theoretical signal responses for our designed microfluidic transmitter and receiver, which further facilitate the transmitter design optimization. Our theoretical results are validated via simulations performed through the COMSOL Multiphysics finite element solver. We demonstrate the predefined nature of the generated pulse and the demodulated rectangular signal together with their dependence on design parameters. 
\end{abstract}

\vspace{-12pt}
\begin{IEEEkeywords}
\vspace{-12pt}	
Molecular communication, microfluidics, microfluidic transmitter, microfluidic receiver, chemical reaction, chemical circuits, genetic circuits.
\end{IEEEkeywords}

%
\IEEEpeerreviewmaketitle

\vspace{-12pt}	
\section{Introduction}
The possibility of harnessing information processing and communication functionalities from physical and chemical processes at the level of molecules has been at the basis of a great bulk of research in recent years on Molecular Communication (MC)~\cite{ifa_paradigm,Akyildiz15,Farsad16}. The physical processes of molecule propagation usually include diffusion and convection, which govern the molecule transport and can usually be described by a convection-diffusion equation \cite{berg1993random,bruustheoretical}. Meanwhile, chemical reactions may occur during molecule propagation via enzyme reaction \cite{6712164}, or at the reception of molecule via reversible absorption reaction \cite{yansha2016stochastic} or ligand binding reaction \cite{alberts2013essential}. To capture the molecule behaviour at any time, existing research has mainly focused on mathematically modelling and theoretical analysis of these physical and chemical processes, such as the channel response modelling \cite{yansha2016stochastic,7511443}, channel capacity calculation \cite{6305481,7541454}, and bit error probability derivation \cite{6712164,8633972}.

Despite substantial research outcomes in the above theoretical study, the design and prototyping of components with MC functionalities has been less explored except from some works \cite{farsad2013tabletop,7397863,giannoukos2018chemical,6630482,6668865,8418677,8255057}, partly because of the highly interdisciplinary technical knowledge and tools required to engineer these systems in practice. Existing MC prototypes can be classified into macroscale MC prototypes \cite{farsad2013tabletop,7397863,giannoukos2018chemical} and nanoscale or microscale MC prototypes \cite{6630482,6668865,8418677,8255057}. The macroscale testbeds in \cite{farsad2013tabletop,7397863,giannoukos2018chemical} considered the information sharing over a distance via alcohol and odor particles, but these macroscale testbeds are inapplicable or inappropriate to be operated in very small dimensions or in specific environment, such as in the water or in the human body. Besides, the detection of signaling molecules heavily relies on electrical devices, including sensors and mass spectrometry (MS), where the signal processing over chemical signals has been less explored in the molecular domain.

For microscale MC testbeds, the authors in \cite{6630482} proposed a Hydrodynamic Controlled Microfluidic Network (HCN) and demonstrated how to realize a pure hydrodynamic microfluidic switching function, where the successful routing of payload droplets was achieved by designing the geometry of microfluidic circuits. In \cite{6668865}, the genetically engineered \textit{Escherichia coli} (\textit{E. coli}) bacteria, housing in a chamber inside a microfluidic device, serves as a MC receiver using fluorescence detection upon the receipt of the signaling molecule C6-HSL. Note that the microfluidic channel in \cite{6668865} was only used as a propagation pathway for C6-HSL molecule, and the authors did not analytically evaluate the response of the C6-HSL molecule transport inside microfluidics. Furthermore, the microfluidic designs in \cite{6630482,6668865} did not realize any signal processing functions, such as modulation and demodulation, in molecular domain.

Signal processing functions performed over electrical signals or devices usually involves a highly complex procedure, and the utilization of electrical devices faces challenges, such as unbiocompatibility and invasiveness, for biomedical-related applications. This motivates us to perform signal processing directly over chemical signals. In general, signal processing functions over chemical signals can be achieved using two approaches: 1) biological circuits \cite{weiss2003genetic} in engineered living cells, and 2) chemical circuits \cite{cook2009programmability} based on “non-living” chemical reactions. 
Existing works in \cite{8418677} have already designed biological circuits to realize the parity-check encoder and decoder. However, the utilization of biological cells for MC currently faces challenges such as slow speed, unreliability, and non-scalability, which motivates our initial work \cite{8255057}. In \cite{8255057}, we designed a chemical reaction-based transmitter for MC, where the transmitter is capable of generating a molecular concentration pulse upon a rectangular triggering signal, thus realizing the modulation function. This work is inspired by how cells generate pulse-shaped molecular signals in biology, and motivated by a bulk of MC literature on information transmission via molecular pulse modulation~\cite{Pehlivanoglu17,yansha2017stochastic}. In this paper, we expand our previous work and make the following contributions:
\begin{itemize}
	\item We first present our designed microfluidic transmitter capable of generating a molecular concentration pulse upon a rectangular triggering signal. To ensure the successfully pulse generation, we define three chemical reactions, where the sequence of each reaction is controlled by the microfluidic channel geometry.  	
	\item We then propose the microfluidic receiver design capable of demodulating a received signal to a rectangular output signal. This demodulation is realized via two chemical reactions, where a thresholding reaction is proposed to first deplete the received signal below the threshold, and an amplifying reaction converts the residual received signal into the output signal.	
	\item  Unlike \cite{6668865}, our microfluidic design is supported by both analytical and numerical simulation results. We derive the channel responses of the straight convection-diffusion-reaction channels with rectangular and Gaussian inlet concentration, which can be reduced to that of straight convection-diffusion channels.
	\item 
	To optimize the system performance, we propose a reaction channel length optimization flow to provide detailed instructions on how to control the maximum concentration of a generated pulse, which provides an insight into the dependence of the maximum concentration on design parameters. In addition, we analyse the restricted time gap between two consecutive input signals to ensure a continuous transmission of non-distorted pulses. Finally, the analytical results are validated against simulations performed in the COMSOL Multiphysics finite element solver.
\end{itemize}

The rest of the paper is organized as follows. In Sec.~\ref{sec:II}, we present the microfluidic transmitter and receiver design in terms of chemical reactions and microfluidic components. 
Sec. \ref{sec:m} introduces microfluidic characteristics and theoretically analyses convection-diffusion channels and convection-diffusion-reaction channels. In Sec. \ref{sec:TX} and \ref{sec:RX}, we not only present the analysis and design for the proposed microfluidic transmitter and receiver, respectively, but also provide numerical simulation results performed in COMSOL Multiphysics. In Sec. \ref{MC}, we combine the microfluidic transmitter with the receiver to show a basic end-to-end MC system. Finally, Sec.~\ref{sec:conclusion} concludes the paper.

\vspace{-12pt}	
\section{System Model}
\label{sec:II}
The overall scheme of the proposed transmitter and receiver for MC is shown in Fig.~\ref{fig:overall_scheme}.
\begin{figure}[!tb]
	\centering
	\includegraphics[width=6.6in]{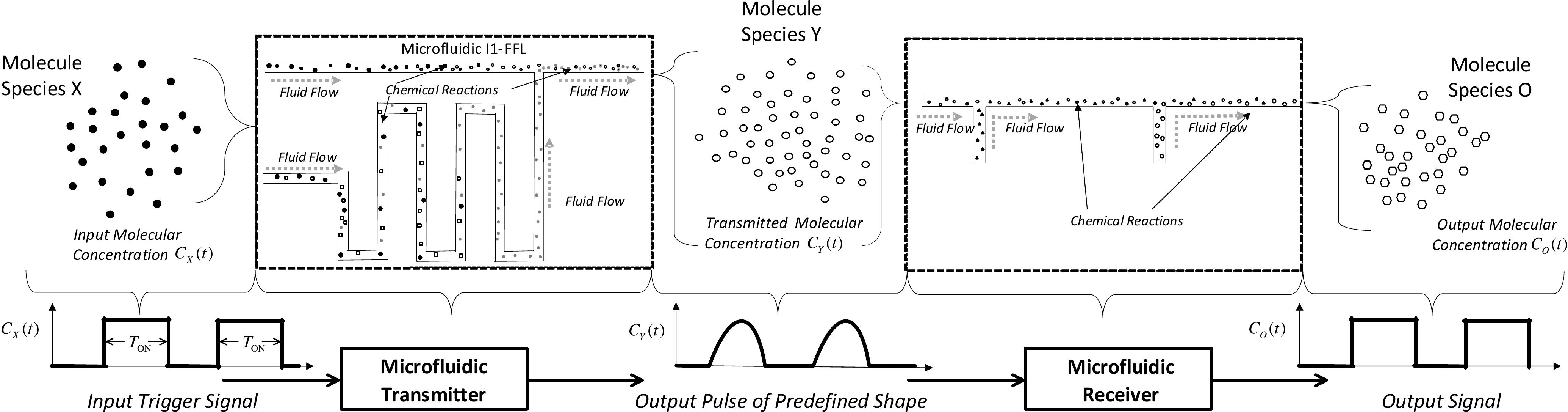}
	\caption{Overall scheme of the proposed transmitter and receiver for MC.}
	\label{fig:overall_scheme}
\end{figure} 
At the microfluidic transmitter, a rectangular input molecular signal composed of the molecular species $X$ in a fluid with concentration $C_X(t)$ 
enters the microfluidic transmitter that upon a variation in $C_X(t)$ produces an output another molecular signal composed of molecular species $Y$ with concentration $C_Y(t)$ by following a predefined pulse shape. After diffusion of emitted pulse $C_Y(t)$, a microfluidic receiver is designed to demodulate the received pulse to a rectangular output signal using species $O$ with concentration $C_O(t)$. Here, both the pulse shape and the demodulated signal shape are dependent on the values of parameters in the microfluidic device implementation. As the fluids flow through microfluidic device channels, a series of chemical reactions occur to generate the molecules of species $Y$ and species $O$, which guarantee the successful pulse generation and the signal demodulation. In the following, we first introduce these chemical reactions at the transmitter side and receiver side, and then describe the microfluidic components of the transmitter and receiver.
\vspace{-10pt}	
\subsection{Chemical Reactions Design for the Microfluidic MC Devices}

\subsubsection{\textbf{Chemical Reactions Design for the Microfluidic Transmitter}}
Gene regulatory networks are sets of interconnected biochemical processes in a biological cell~\cite{karlebach2008modelling}, where DNA genes are linked together by activation and repression mechanisms of  certain biological macromolecules that regulate their expressions into proteins. Each DNA gene contains coding sequences and regulatory sequences, which are sites the proteins (transcription factor) can bind and control the rate of the gene expression, either by increasing (activation) or decreasing (repression) the rate of protein synthesis. In gene regulatory networks, genes are interconnected such that the proteins produced by one or more genes regulate the expression of one or more genes, which results in complex protein expression dynamics.

Gene regulatory networks can be abstracted with nodes representing the genes, interconnected by directed edges that correspond to the control of a gene (edge destination) expression by a transcription factor encoded by another gene (edge source). Network motifs are patterns of nodes and directed edges that occur more frequently in natural gene transcription networks than randomized networks~\cite{milo2002network}.  The Feed Forward Loop (FFL) is a family of network motifs among all three-node patterns frequently observed in nature~\cite{milo2002network,alon2007network}. In the structure of FFL, the transcription factor protein $X$ regulates the genes expressing other two proteins, namely, $P$ and $Y$, where $P$ is also a transcription factor that regulates the gene expressing protein $Y$. Depending on the types of these regulations, either activation or repression, there are 8 different FFLs~\cite{alon2006introduction}.

\begin{figure}[t]
	\centering
	\includegraphics[width=0.5in]{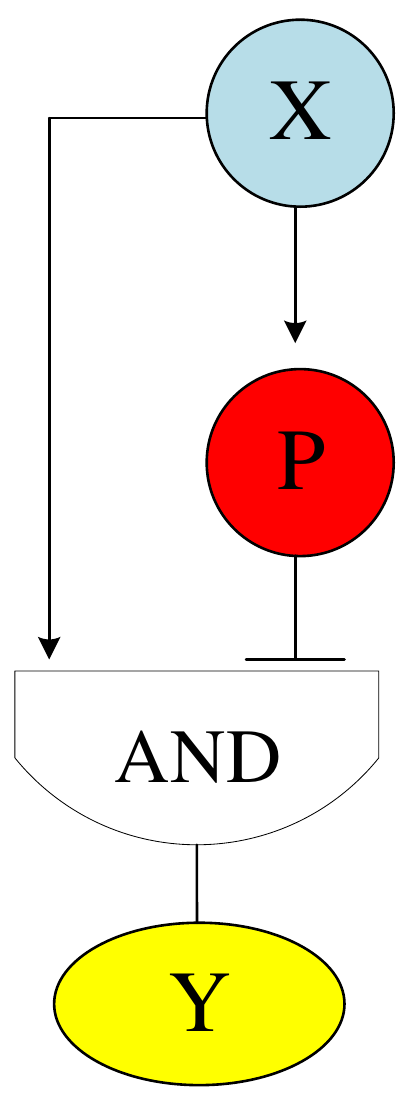}
	\caption{The I1-FFL network motif.}
	\label{fig:network_motif}
\end{figure} 
Among all the FFLs found in nature, the I1-FFL results in a pulse-like dynamics of its output $Y$~\cite{alon2007network}. As shown in Fig. \ref{fig:network_motif},
an input gene expresses the protein $X$, which is a transcription factor for the genes expressing $Y$ and $P$. In presence of $X$, the expressions of the genes encoding protein $Y$ and protein $P$ are activated, resulting in the build up of the concentrations of protein $Y$ and protein $P$, respectively. On its turn, the protein $P$ is another transcription factor that works as a repressor for the gene encoding protein $Y$. The AND input to the gene that encodes $Y$ corresponds to a situation where this gene is activated when the transcription factor $X$ binds to the regulatory sequence, but it is inactivated whenever transcription factor $P$ binds to the same sequence independently from the presence of $X$. In such a way, protein $X$ initializes the rapid expression of the gene encoding protein $Y$ first, and after a delay, enough $P$ accumulates and represses the production of protein $Y$, whose concentration will continuously decrease because of natural degradation. This generates a pulse shape for the concentration of protein $Y$ as a function of the time.

One example of I1-FFL is the galactose system of \textit{E. coli}, where the galactose utilization operon (a cluster of genes sharing the same regulatory sequences and expressed together) $gal$ETK is regulated in an I1-FFL by the activator $CRP$ ($X$), and the repressor $galS$ ($P$) \cite{mangan2006incoherent}. Results showed that in nature we can observe a pulse-like expression of the $gal$ETK genes, which is initiated by a step variation of active $CRP$ mediated by the molecular species $cAMP$.

In this paper, we take inspiration from the I1-FFL to design a transmitter in the molecular domain. Although the discipline of synthetic biology is opening the road to the programming of functionalities in the biochemical environment through genetic engineering of biological cells~\cite{Kahl13}, there are a number of factors that encouraged an alternative technology for the design of a MC transmitter in this paper, such as the small number of molecules involved for each cell together with difficulties in coordinating multiple cells, the added complexity of cellular behavior, including cell growth, evolution, and biological noise, and the slow response time of genetic regulatory networks such as the I1-FFL, whose output pulse shape is usually realized in nature in the order of cell generation time (hours) as indicated in~\cite[Fig. 4]{mangan2006incoherent}.

Inspired by the I1-FFL mechanism in gene regulation networks, we explore the realization of I1-FFL via mass action chemical reactions, \textit{i.e.},  processes that convert one or more input molecules  (\textit{reactants}) into one or more output molecules (\textit{products}). Reactions may proceed in forward or reverse directions, which are characterized by forward ($k_f$) and reverse ($k_r$) reaction rates, respectively. Within the scope of this paper, we assume unbalanced reactions where the forward reaction rate is much greater than the reverse rate. A chemical reaction network is defined as a finite set of  reactions involving a finite number of species~\cite{cook2009programmability}, where these reactions occur in a well-stirred environment, aiming to realize a function or algorithm via mass action chemical reactions. Specific chemical reaction networks have already been designed for signal restoration, noise filtering, and finite automata, respectively, through a discipline known as molecular programming~\cite{Klinge16modu}.

To execute the same functionality of an I1-FFL with a chemical reaction network, we define three chemical reactions as follows:
\begin{align}
\text{Reaction I}: X + {S_y} \to Y, \label{CR1}\\  
\text{Reaction II}: X + {S_p} \to P,\label{CR2} \\ 
\text{and} \;\;\text{Reaction III}:   Y +P\to Z,  \label{CR3}
\end{align}
where these reactions involve the input molecular species $X$, the molecular species ${S_p}$ and ${S_y}$, the intermediate product molecular species $P$, and the output molecular species $Y$.

In the I1-FFL gene regulation network, the active $X$ first activates the gene expressing the protein $Y$, and
only when $P$ accumulates sufficiently, it suppresses the expression of the protein $Y$, generating the aforementioned pulse-like concentration signal. Here, the molecular species  $X$, ${S_p}$, and ${S_y}$ are only injected at $t=0$, and the chemical reactions in \eqref{CR1}, \eqref{CR2}, and \eqref{CR3} happen simultaneously with a much quicker speed under well-stirred environment than that of the I1-FFL gene regulation network dynamics, which may not result in the pulse-like output signal $Y$ when these three reactions have the same reaction rate. One way to cope with it is to adjust the reaction rate to be different among these reactions.

However, in practice, we want to design the molecular communication system with the pulse-like output triggered by the rectangular pulse input representing bit-1 transmission. In such a way, the output pulse only occurs inside the duration of a rectangular pulse input, and every bits are modulated to their corresponding pulses as shown in Fig. \ref{fig:overall_scheme}. To control the rectangular pulse input signals,  the sequence of each reaction, and the delayed arrival of product $P$ after Reaction II in~\eqref{CR2},  we propose a microfluidic transmitter to realize the same functionality of I1-FFL as in gene regulation network in Fig.~\ref{bf2} and containing the reactions \eqref{CR1}, \eqref{CR2}, and \eqref{CR3}.

\subsubsection{\textbf{Chemical Reactions Design for the Microfluidic Receiver}} \label{receiver_reaction}
According to the demodulation requirement of traditional communication systems, we aim to design a microfluidic receiver capable of demodulating the received pulse to a rectangular signal. To do so, we design the chemical reactions as follows:
\begin{align}
\text{Reaction IV}: Y + ThL \to Waste, \label{CR4}  \\
\text{and} \;\;\text{Reaction V}: Y+Amp {\to} Y+O,\label{CR5} 
\end{align}
where these reactions involve the input molecular species $Y$, the molecular species $ThL$ and $Amp$, intermediate product molecular species $Waste$, and the output molecular species $O$. Once the species $Y$ arrives at the receiver, the $\text{Reaction IV}$ is immediately activated, resulting in a depletion of species $Y$ that is below the concentration of species $ThL$. Then, any remaining $Y$ catalyses the conversion of species $Amp$ into the output species $O$. Obviously, output species $O$ will only be produced when the concentration of $Y$ is greater than the concentration of $ThL$, so we regard the concentration of $ThL$ as a threshold and name $\text{Reaction IV}$ as the thresholding reaction. $\text{Reaction V}$ refers to an amplifying reaction. Similar to the chemical reactions at the transmitter, the sequence of $\text{Reaction IV}$ and $\text{Reaction V}$ is controlled by the microfluidic receiver geometry design, which will be presented next.
\vspace{-12pt}	
\subsection{Microfluidic Device Design}
In this subsection, we describe each component of our proposed microfluidic transmitter and receiver, in Fig.~\ref{bf2}. 
	\begin{figure}[!tb]
	\centering
	\includegraphics[width=6.5in]{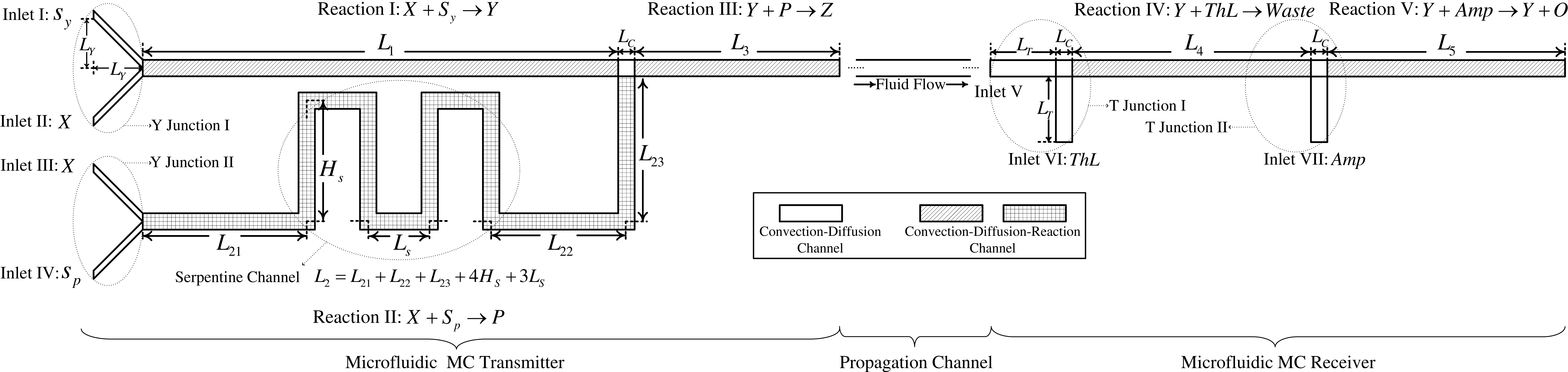}
	\caption{Novel Design of the microfluidic MC transmitter and receiver.}
	\label{bf2}
	\end{figure} 
A microfluidic device is a system that can process or manipulate small ($10^{-9}$ to $10^{-18}$ litres) amount of fluids using channels in dimensions of tens to hundreds of micrometres \cite{whitesides2006origins}. Recently, an increasing number of biological and chemical experiments are conducted in microfluidic or lab-on-a-chip (LOC) devices, due to inherent advantages in miniaturization, integration, portability and automation with low reagents consumption, rapid analysis, and high efficiency \cite{li2013microfluidic}. According to whether a chemical reaction occurs in a microfluidic channel, we classify microfluidic components as two types: 1) convection-diffusion channel, and 2) convection-diffusion-reaction channel. 
\subsubsection{\textbf{Convection-Diffusion Channel}}
\begin{itemize}
	\item \textbf{Y Junction at the microfluidic transmitter:} 
	The reactions between reactants require mixing to occur in a short distance, which can be facilitated by diffusion in Y junctions. Y junctions are configured by one outlet and two inlets, \textit{i.e.}, Y junction I and Y junction II in Fig.~\ref{bf2}, where the outlet width is doubled compared with each inlet width, and the angle between the main channel and the first inlet starting anticlockwise from the main channel is $145^o$. The fluid flow containing input reactant $X$ with concentration $C_{X_0}^{\text{II}}$ and $C_{X_0}^{\text{III}}$ is injected into the Inlet II and Inlet III using syringe pumps, which can be described by a rectangular pulse signal, as in Fig.~\ref{fig:overall_scheme}, with the width equalling to the length of injection time $T_{\rm{ON}}$, whereas the reactant $S_y$ with concentration $C_{S_{y_0}}^{\text{I}}$ and reactant $S_p$ with concentration $C_{S_{p_0}}^{\text{IV}}$ are continuously injected into Inlet I and Inlet IV, respectively. By doing so, the flows from Inlet I and Inlet IV can flush the microfluidic device continuously without influencing Reaction III in~\eqref{CR3}.	
	\item \textbf{T Junction at the microfluidic receiver:} T junctions are chosen at the receiver 
	equipping with the same functionality as Y junctions. A T Junction has one outlet and two inlets, \textit{i.e.}, T junction I and T junction II in Fig.~\ref{bf2}, where the angle between the second inlet starting anticlockwise from the first inlet is $90^o$, and one inlet of T junction II is merged into a convection-diffusion-reaction channel. After diffusion, the transmitted molecules from microfluidic transmitter propagate to enter the receiver, and the reactant $ThL$ with concentration $C_{ThL}^{\text{VI}}$ and $Amp$ with concentration $C_{Amp}^{\text{VII}}$ are continuously injected into the Inlet VI and Inlet VII, respectively.
	\item \textbf{Straight Convection-Diffusion Channel:} This channel is used to connect the transmitter with the receiver and provides a propagation pathway for a generated pulse. 
\end{itemize}

\subsubsection{\textbf{Convection-Diffusion-Reaction Channel}}
For simplicity, in the following, 
we refer to the channel in which Reaction $i$ happens as the Reaction $i$ channel.
	\begin{itemize}
			\item \textbf{Transmitter}
				\subitem \textbf{Straight Reaction I channel:} The outflow of Y junction I passes through the Reaction I channel with length $L_1$ to realize the Reaction I in \eqref{CR1} to generate the output signal $Y$.
				\subitem \textbf{Serpentine Reaction II channel:} The outflow of Y junction II passes through the Reaction II channel to generate $P$ according to the Reaction II in \eqref{CR2}. To realize the pulse-shaped concentration of emitted signal $Y$, the Reaction II channel is designed to be longer than the Reaction I channel, with the result of delaying the contact between species $P$ and $Y$, and therefore delaying the Reaction III. Furthermore, a serpentine channel is designed and replaced a straight reaction channel to delay the arrival of species $P$ in a compact space within the microfluidic transmitter. The width and height of the serpentine channel is denoted as $L_s$ and $H_s$, respectively. The design in Fig.~\ref{bf2} is conventionally denoted as containing 2 delay lines, due to its two bended tubes with height $H_s$ in the serpentine channel. The equivalent straight channel length of this serpentine channel is denoted as $L_2$ and can be calculated as $L_2=L_{21}+L_{22}+L_{23}+4H_{s}+3L_{s}$.
				\subitem \textbf{Straight Reaction III channel:} Once $P$ arrives at the Reaction III channel with length $L_{3} $, Reaction III in \eqref{CR3} occurs to decrease the output signal $Y$.

			\item \textbf{Receiver} 
			\subitem \textbf{Straight Reaction IV channel:} The outflow of T junction I flows through the Reaction IV channel with length $L_4$ to deplete $Y$ below the concentration of species $ThL$ according to Reaction IV in \eqref{CR4}.
			\subitem \textbf{Straight Reaction V channel:} When the remaining $Y$ arrives at the Reaction V channel with length $L_5$, Reaction V in \eqref{CR5} is activated to convert the species $Amp$ into output species $O$.

	\end{itemize}

\vspace{-12pt}	
\section{Basic Microfluidic Channel Analysis}
\label{sec:m}

In this section, we first describe the basic characteristics of microfluidics, and then use 1D model to approximate and derive analytical expressions for convection-diffusion channels and convection-diffusion-reaction channels. Numerical results are provided to verify our theoretical analysis.

\vspace{-10pt}	
\subsection{Basic Characteristics of Microfluidics}
The nature of the flow highly depends on the Reynolds
number, which is the most famous dimensionless parameter
in fluid mechanics. For flow in a pipe, the Reynolds number
is defined as \cite{bruustheoretical}
\begin{equation}
	\text{Re}=\frac{\rho v_{\text{eff}}D_H}{\mu},
\end{equation}
where $\rho$ is the fluid density, $v_{\text{eff}}$ is the fluid mean velocity, $D_H$ is the hydraulic diameter of the channel, and $\mu$ is the constant fluid viscosity. When we scale down standard laboratory channels from decimeter scale to microscopic scale, Reynolds number is usually very small (Re $<1$), which indicates that flows become laminar flows, such that an ordered and regular streamline pattern can be experimentally observed \cite{di2009inertial}. Applying a long, straight, and rigid microfluidic channel  to a flow and imposing a pressure difference between the two ends of the channel, the flow is referred to as the Poiseuille flow \cite{bruustheoretical}.
When the cross section of the microfluidic channel is circle-shaped , the flow velocity profile can be described as
\begin{align}\label{velocity_profile}
	v(r)=2v_{\text{eff}}(1-\frac{r^2}{R^2}),
\end{align}
where $r$ is the radial distance, and $R$ is the radius of the cross section.

 \vspace{-10pt}	
\subsection{Convection-Diffusion Channels}
\label{sec:convection_diffusion}
For one type of molecular species flowing in a 3D straight convection-diffusion channel with rectangular cross section whose height is $h$ and width is $w$, its concentration $C(x,y,z,t)$ can be described by the 3D convection-diffusion equation as \cite{stocker2011introduction}
\begin{equation}
\frac{\partial C(x,y,z,t)}{{\partial}t}=D \nabla^2 C(x,y,z,t)-\boldsymbol{v}
\cdot \nabla C(x,y,z,t),
\label{b1}
\end{equation}
where $\nabla$ is the Nabla operator, and $\boldsymbol{v}$ is the flow velocity that can be solved by Navier-Stokes equation \cite{bruustheoretical}. When the flow falls into dispersion regime, the interaction between cross-sectional diffusion and non-uniform convection can lead to an uniform molecule distribution along the cross-section, \textit{i.e.}, $\frac{\partial C(x,y,z,t)}{{\partial}y}=\frac{\partial C(x,y,z,t)}{{\partial}z}=0$, such that (\ref{b1}) can be simplified into a 1D convection-diffusion equation \cite{wicke2018modeling}
\begin{equation}
\frac{\partial C(x,t)}{{\partial}t}=D_{\text{eff}}\frac{\partial^2 C(x,t)}{{\partial}x^2}-v_{\text{eff}}\frac{\partial C(x,t)}{{\partial}x},
\label{b2}
\end{equation}
where $D_{\text{eff}}=(1+\frac{8.5{v_{\text{eff}}^2}{h^2}{w^2}}{210D^2({h^2}+2.4hw+{w^2})})$ is the \textit{Taylor-Aris} effective diffusion coefficient \cite{bicen2014end}. 

\vspace{-10pt}	
\subsection{Convection-Diffusion-Reaction Channels}
\label{tx_reaction_channel}

Unlike a convection-diffusion channel, the molecular transport is not only affected by convection-diffusion, but also affected by reactions in a reaction channel. To quantitatively describe the chemical reaction and dispersion of molecules at a straight microfluidic channel, we introduce the 1D convection-diffusion-reaction equation. For a general reaction $A+B\to AB$, the spatial-temporal concentration distribution of species $A$ and $AB$ can be described as
\begin{align}
& \frac{{\partial {C_A}(x,t)}}{{\partial t}}={D_\text{eff}}\frac{{{\partial ^2}{C_A}(x,t)}}{{\partial {x^2}}}- {v_\text{eff}}\frac{{\partial {C_A}(x,t)}}{{\partial x}}-k{C_A}(x,t){C_{{B}}(x,t)}  , \label{reaction1}
\\& \frac{{\partial {C_{AB}}(x,t)}}{{\partial t}}= \underbrace{{D_\text{eff}}\frac{{{\partial ^2}{C_{AB}}(x,t)}}{{\partial {x^2}}}}_{\text{Diffusion}} - \underbrace{{v_\text{eff}}\frac{{\partial {C_{AB}}(x,t)}}{{\partial x}}}_{\text{Convection}} + \underbrace{{k}{C_A}(x,t){C_{{B}}(x,t)}}_{\text{Reaction}}, \label{reaction1Y}
\end{align}
where $k$ is the rate constant. Assuming species $B$ with concentration $C_{B_0}$ is continuously injected at the inlet of the channel at $x=0$ and $t=0$ with velocity $v_\text{eff}$, we solve the above convection-diffusion-reaction equations in the following two theorems when species $A$ is injected with a rectangular concentration profile and a Gaussian concentration profile.
\vspace{-10pt}	
\begin{theorem}
	\label{the1}
	With species $A$ following a rectangular concentration distribution 
	\begin{align}
		C_{A}(0,t)=C_{A_0}[u(t)-u(t-T_\text{ON})]
	\end{align}
	being injected at the inlet of a straight microfluidic channel at $x=0$ and $t=0$ using velocity $v_\text{eff}$, the concentration distributions of $A$ and $AB$ are derived as 
	\begin{equation}\label{b4}
	{C_A}(x,t)=\begin{cases}
	g(x,t), &0\leq t \leq T_\text{ON}\\
	g(x,t)-g(x,t-T_\text{ON}), &t> T_\text{ON},
	\end{cases}
	\end{equation}
	and
	\begin{equation}\label{b5}
	{C_{AB}}(x,t)=\begin{cases}
	h(x,t)-g(x,t), &0\leq t \leq T_\text{ON}\\
	[h(x,t)-g(x,t)]-[h(x,t-T_\text{ON})-g(x,t-T_\text{ON})], &t> T_\text{ON},
	\end{cases}
	\end{equation}
	where 
	$g(x,t) = \frac{C_0}{2}\left\{ {\text{exp} \left[ {\frac{{\left( {{v_\text{eff}} - \alpha } \right)x}}{{2{D_\text{eff}}}}} \right]\text{erfc}\left[ {\frac{{x - \alpha t}}{{2\sqrt {{D_\text{eff}}t} }}} \right]} \right.
	\nonumber  \left. { + \text{exp} \left[ {\frac{{\left( {{v_\text{eff}} + \alpha } \right)x}}{{2{D_\text{eff}}}}} \right]\text{erfc}\left[ {\frac{{x + \alpha t}}{{2\sqrt {{D_\text{eff}}t} }}} \right]} \right\}$,
	$h(x,t)=\frac{C_{A_0}}{2}[\text{erfc}(\frac{x-{v_\text{eff}}t}{2\sqrt{D_{\text{eff}}}t})+ e^{\frac{{v_\text{eff}}x}{D_{\text{eff}}}}   \text{erfc}(\frac{x+{v_\text{eff}}t}{2\sqrt{D_{\text{eff}}}t})]$
	with ${C_{0}}=\min\left\{C_{A_0},C_{B_0} \right\}$ and $\alpha  = \sqrt {{{v_\text{eff}}^2} + 4{k}{C_{0}}{D_\text{eff}}}  $.
\end{theorem}
\begin{proof}
	See the Appendix \ref{A}.
\end{proof}
\vspace{-10pt}	
\begin{theorem}
	\label{the2}
	With species $A$ following a Gaussian concentration distribution 
	\begin{align}
	C_{A}(0,t)=\frac{C_{A_0}^{\text{}}}{\sqrt{2 \pi \sigma^2}} e^{-\frac{(t-\mu)^2}{2\sigma^2}}
	\end{align}
	being injected at the inlet of a straight microfluidic channel at $x=0$ and $t=0$ using velocity $v_\text{eff}$ and $C_{B_0}>\max\left\{ C_A(0,t) \right\}$, the concentration distribution of $A$ can be approximated as 
	\begin{equation}\label{approx1}
	{C_A}(x,t) \approx C_{A}^{\text{Appro$_1$}}(x,t)=\begin{cases}C_A(0,t-\frac{x}{v_\text{eff}})-C_{B_0}, &t_1+\frac{x}{v_\text{eff}}\leq t \leq t_2+\frac{x}{v_\text{eff}},\\
	0, &\text{otherwise}.	
	\end{cases}
	\end{equation}	
	\begin{equation} \label{approx2}
	\begin{aligned}
		\text{or}~{C_A}(x,t) \approx C_{A}^{\text{Appro$_2$}}(x,t)=\frac{1}{2 \pi}\int_0^\infty [e^{-j\omega t}\overline{\widetilde{{C_A^{\text{Appro$_2$}}}}(x,\omega)}+e^{j\omega t}{\widetilde{{C_A^{\text{Appro$_2$}}}}(x,\omega)}]\mathrm{d}w,
	\end{aligned}
	\end{equation}	
	where
	\begin{align}
		\widetilde{{C_A^{\text{Appro$_2$}}}}(x,s)=l(s)e^{\frac{v_\text{eff}-\sqrt{{v_\text{eff}}^2+4D_{\text{eff}}s}}{2D_{\text{eff}}}x},
	\end{align}
	\begin{align}
		l(s)={C_{A_0}^{\text{}}}e^{-s \mu+\frac{{(\sigma s)}^2}{2}}[Q({\frac{t_1+\sigma^2 s-\mu}{\sigma}})-Q({\frac{t_2+\sigma^2 s-\mu}{\sigma}})]-\frac{C_{B_0}^{\text{}}}{s}(e^{-st_1}-e^{-st_2}),
	\end{align}
	\begin{align}
		t_1 =\mu -\sqrt{-2\sigma^2 \ln \frac{C_{B_0}^{\text{}}\sqrt{2\pi \sigma^2}}{C_{A_0}^{\text{}}}}, \\
		and ~t_2 =\mu +\sqrt{-2\sigma^2 \ln \frac{C_{B_0}^{\text{}}\sqrt{2\pi \sigma^2}}{C_{A_0}^{\text{}}}}.
	\end{align}
	
\end{theorem}
\begin{proof}
	See the Appendix \ref{B}.
\end{proof}
Our result $C_{A}^{\text{Appro$_2$}}(x,t)$ can be easily computed using Matlab. Importantly, \eqref{b4}, \eqref{approx1}, and \eqref{approx2} reduce to solutions of a convection-diffusion equation when $C_{B_0}=0$.

In Fig. \ref{reaction_rectangular} and \ref{reaction_gaussian}, we plot the analytical outlet concentrations of species $AB$ in \textbf{Theorem 1}, species $A$ in \textbf{Theorem 2} and their simulation results using COMSOL, where we use “Ana.” and “Sim.” to abbreviate “Analytical” and “Simulation”, respectively, and this notation is also used throughout the rest of this paper.
\begin{figure}[tb]
	\centering
	\begin{minipage}[t]{0.48\textwidth}
		\centering
		\includegraphics[width=3in]{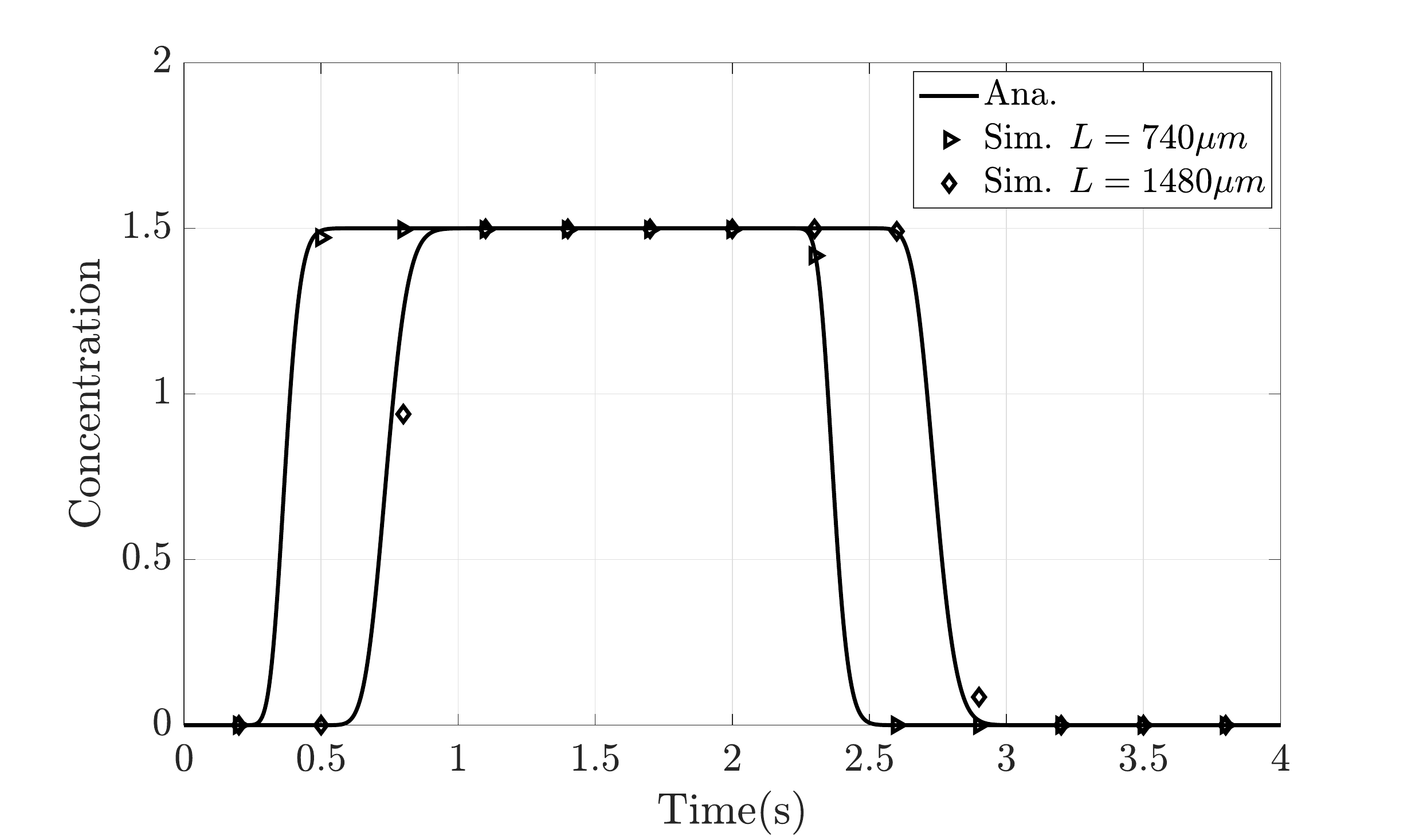}
		\caption{The concentration of species $AB$ in \textbf{Theorem 1} with different channel length $L$.}
		\label{reaction_rectangular}
	\end{minipage}
	\quad
	\begin{minipage}[t]{0.48\textwidth}
		\centering
		\includegraphics[width=3in]{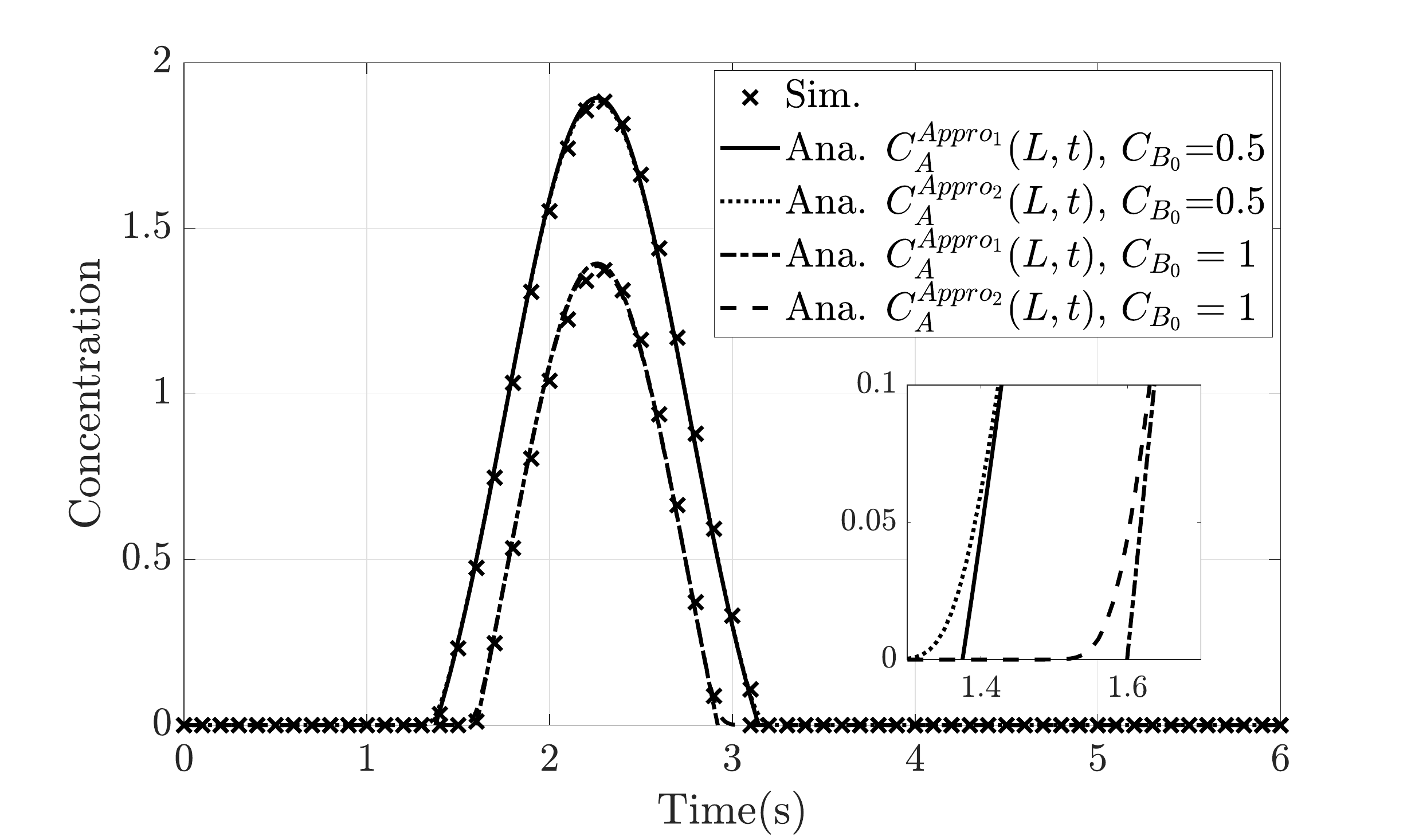}
		\caption{The concentration of species $A$ in \textbf{Theorem 2} with $L=540\mu$m and different $C_{B_0}$.}
		\label{reaction_gaussian}
	\end{minipage}	
\end{figure}
We set the parameters: $C_{A_0}=C_{{B_0}}=1.5$mol/m$^3$, $C_{A_0}=3$mol/m$^3$, $\mu=2$, $\sigma^2=0.25$, $D_\text{eff}=10^{-8}$m$^2$/s, $k=400$m$^3$/(mol$\cdot$s), $T_\text{ON}=2$s. The simulation points are plotted using the outlet of a straight microfluidic channel with rectangular-shaped cross section, $h= 10\mu$m and $w= 20\mu$m, where the species $A$ and $B$ are both injected with the same velocity $v_\text{eff}=0.2$cm/s. In Fig. \ref{reaction_rectangular}, it clearly demonstrates a close match between the analytical curves and the simulation points with different channel length $L$. In Fig. \ref{reaction_gaussian}, we observe that both approximation methods capture the residual concentration variation of $A$ after reaction $A+B \to AB$. When $C_A$ approaches to zero, the curve using the second approximation method is smoother than that using the first approximation method due to the consideration of diffusion effect.

\vspace{-12pt}	
\section{Microfluidic MC Transmitter Analysis and Design Optimization}
\label{sec:TX}
In this section, we first analyse the Y Junction and three reaction channels, and then we provide the microfluidic transmitter design in terms of the optimal design of the reaction II channel length and the restricted time gap between two consecutive input bits, which enable us to control the maximum concentration of a generated pulse and ensure a continuous transmission of non-distorted pulses, respectively.
\vspace{-12pt}	
\subsection{Microfluidic MC Transmitter Analysis}

\subsubsection{\textbf{Y Junction}}
The fluid flow containing input reactant $X$ with concentration 
\begin{align}
	C_{X}^{\text{II}}(x,t)=C_{X_0}^{\text{II}}[u(t)-u(t-T_\text{ON})]\\
	\text{and}~C_{X}^{\text{III}}(x,t)=C_{X_0}^{\text{III}}[u(t)-u(t-T_\text{ON})]
\end{align}
is injected into Inlet II and Inlet III using syringe pumps, whereas the reactant $S_y$ with concentration $C_{S_{y_0}}^{\text{I}}$ and reactant $S_p$ with concentration $C_{S_{p_0}}^{\text{VI}}$ are continuously injected into Inlet I and Inlet IV, respectively. We let the inlets of a Y Junction as the location origin ($x=0$) and let the time that species are injected at Y Junction inlets as the time origin ($t=0$). For Y Junction I, the outlet concentration of species $X$ can be expressed using \eqref{b4} in \textbf{Theorem 1} with $C_{B_0}=0$ and a substitution of $C_{X_0}^{\text{II}}$ for $C_{A_0}$. However, the complicated form of \eqref{b4} will make Reaction I channel intractable since the outlet concentration of species $X$ at Y Junction I is an initial boundary condition for the convection-diffusion-reaction equation describing Reaction I channel. Take into account that the Y Junction length is shorter than the Reaction I channel length, for simplicity, we assume the outlet concentration of species $X$ is only a time shift of its injected concentration due to the travelling of Y Junction I, that is
\begin{align}\label{yjunction_assumption} 
C_{X}(L_Y,t)\approx C_{X_0}^{\text{II}}[u(t-t_\text{Y})-u(t-T_\text{ON}-t_\text{Y})],
\end{align}
where $t_\text{Y}=\frac{\sqrt{2}L_Y}{v_\text{eff}}$ is the travelling time of a Y Junction ($L_Y$ is marked in Fig. \ref{bf2}). Apparently, the above analysis can also be applied to Y junction II.

\subsubsection{\textbf{Straight Reaction I Channel}}
The outflow of Y junction I enters Reaction I channel to activate $\text{Reaction I}$ in \eqref{CR1}. 
The simultaneous flush of independent $X$ and $S_y$ leads to a concentration dilution, which can be treated as diluting species $X$ using $S_y$ or diluting species $S_y$ using $X$. Hence, with the assumption of \eqref{yjunction_assumption}, the concentration of species $X$ and $S_y$ at the inlet of Reaction I channel become $\frac{1}{2}C_{X}(L_Y,t)$ and $\frac{1}{2}C_{S_{y_0}}^{\text{I}}$, respectively. Based on this, the outlet concentration of species $Y$ can be expressed using \eqref{b5} in \textbf{Theorem 1} by substituting $C_{A_0}$ and $C_{B_0}$ with $C_{X_0}^{\text{II}}$ and $C_{S_{y_0}}^{\text{I}}$, that is   
\begin{align} \label{straight_reaction1}
	C_Y(L_Y+L_1,t)\approx\frac{1}{2}C_{AB}(L_1,t-t_\text{Y}).
\end{align}

Fig. \ref{reaction1_channel} plots the concentration of species $Y$ at Reaction I channel outlet with Y Junction I.
\begin{figure}
		\centering
		\includegraphics[width=3in]{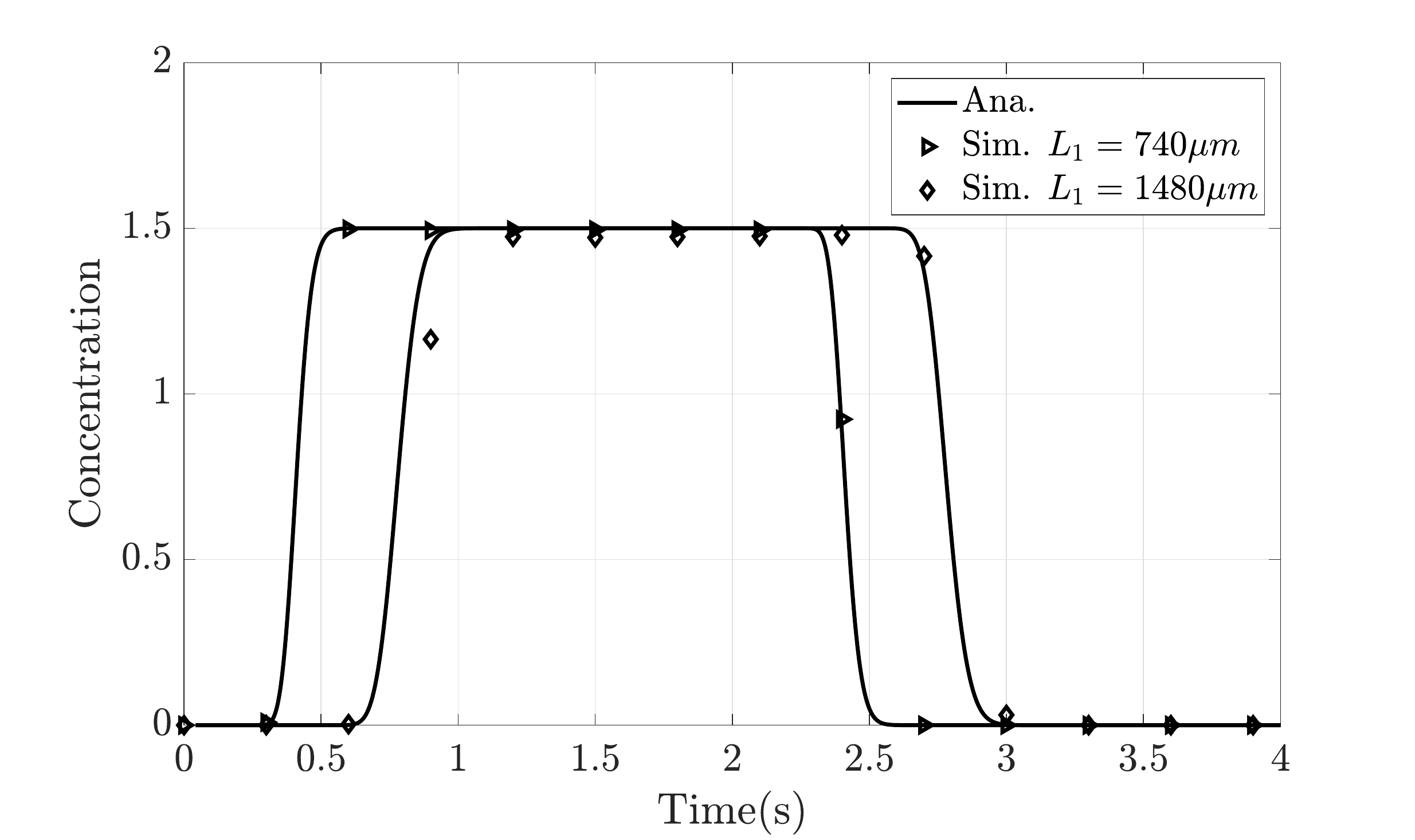}
		\caption{The concentration of species $Y$ at Reaction I channel outlet with Y Junction I.}
		\label{reaction1_channel}
\end{figure}
We set the parameters: $C_{X_0}=C_{{Y_0}}=3$mol/m$^3$, $D_\text{eff}=10^{-8}$m$^2$/s, $k=400$m$^3$/(mol$\cdot$s), $T_\text{ON}=2$s, $v_\text{eff}=0.2$cm/s, $L_Y=60\mu$m, $h= 10\mu$m and $w= 10\mu$m. It is evident that simulation points are in agreement with theoretical analysis in \eqref{straight_reaction1} under different $L_1$, which validates the analysis of straight Reaction I channel.

\subsubsection{\textbf{Serpentine Reaction II Channel}}
The analysis of straight Reaction I channel can also be applied to serpentine Reaction II channel, which yields
\begin{align} \label{serpentine_reaction2}
C_P(L_Y+L_1,t)\approx\frac{1}{2}C_{AB}(L_2,t-t_\text{Y}).
\end{align}  
This can be explained by the following reasons: 1) although turning corners in the serpentine channel usually cause different laminar flows propagating different distances, we can approximate outlet concentrations of the serpentine channel as those of a straight channel with equivalent length when fluids are in low Reynolds number with very small side length tube, and 2) the form of the convection-diffusion-reaction equation and its initial boundary conditions stills hold with only a substitution of $C_P(x,t)$, $C_{S_{p_0}}^{\text{IV}}$, and $C_{X_0}^{\text{III}}$ for $C_Y(x,t)$, $C_{S_{y_0}}^{\text{I}}$, and $C_{X_0}^{\text{II}}$, respectively.

\subsubsection{\textbf{Straight Reaction III Channel}}
The generated species $Y$ and $P$ mix with each other at a conjunction with length $L_C$ and leads to a concentration dilution before flowing to the Reaction III channel. Therefore, at the inlet of straight Reaction III channel, the concentrations of species $Y$ and $P$ are
\begin{align}
	C_Y(L_Y+L_1+L_C,t)\approx\frac{1}{4}C_{AB}(L_1,t-t_\text{Y}-t_\text{C}), \label{Y_r3}\\
	\text{and}~C_P(L_Y+L_1+L_C,t)\approx\frac{1}{4}C_{AB}(L_2,t-t_\text{Y}-t_\text{C}), \label{P_r3}
\end{align}
where  $t_\text{C}=\frac{L_C}{v_\text{eff}}$ is the travelling time of the conjunction. When both species $Y$ and $P$ appear in Reaction III channel, Reaction III in \eqref{CR3} is activated, and the corresponding convection-diffusion-reaction equations can be constructed as \eqref{reaction1} and \eqref{reaction1Y}. Unfortunately, it is foreseeable that deriving the spatial-temporal concentration distribution of species $Y$, exactly the distribution of the generated pulse, is intractable, since the initial condition with the form of $C_{AB}$ in \eqref{b5} is too complicated. However, it is possible to obtain the maximum concentration of the generated pulse, which will be presented in the next subsection.

\vspace{-12pt}	
\subsection{Microfluidic MC Transmitter Design}
\subsubsection{\textbf{Optimal Design of the Reaction II Channel Length}}
\label{optimization1}

As stated earlier, the maximum concentration of a generated pulse, denoted as $\max \left\{C_{\text{TX}} \right\}$, can be obtained, although the convection-diffusion-reaction equation describing Reaction III channel cannot be theoretically solved. In fact, there are many factors affecting $\max \left\{C_{\text{TX}} \right\}$, such as the rate constant $k$ and reaction channel lengths $L_1$, $L_2$, and $L_3$. However, if we assume that the rate constant $k$ and reaction channel lengths collectively ensure that reactants are fully converted into a product in each reaction, the Reaction II channel length $L_2$ will be the only parameter affecting $\max \left\{C_{\text{TX}} \right\}$.
\begin{figure}
	\centering
	\includegraphics[width=2.7in]{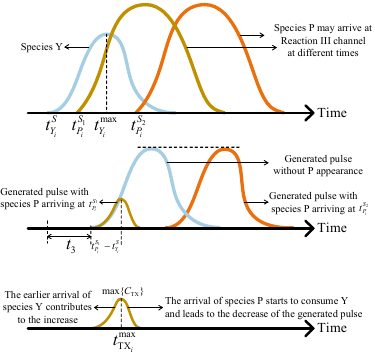}
	\caption{The generated pulses with different arriving time of species $P$ at Reaction III channel. $t_3$ is the travelling time over Reaction III channel.}
	\label{design_flow_r}
\end{figure}

At the transmitter, the design of channel length $L_2>L_1$ allows species $Y$ to first enter the Reaction III channel 
with a result of the concentration increase of a generated pulse, while the late arrival of species $P$ prevents this increase, and leads to a decrease of the generated pulse, as $Y$ will be immediately depleted by $P$ as soon as $P$ appears in Reaction III channel (shown in Fig. \ref{design_flow_r}). Let us denote the arriving and leaving time of a general species $A$ at Reaction III channel inlet as $t_{A_i}^{S}$ and $t_{A_i}^{E}$ for the $i$th input bit, and the time that species $A$ reaches its maximum concentration at Reaction III channel inlet as $t_{A_i}^{\text{max}}$.
There are two situations that lead to different $\max \left\{C_{\text{TX}} \right\}$. 
\begin{itemize}
	\item If $t_{P_i}^{S}<t_{Y_i}^{\text{max}}$, the generated pulse will be consumed by $P$ before reaching $\max  \{C_Y(L_Y+L_1+L_C,t) \}$, causing $\max \left\{C_{\text{TX}} \right\} < \max \left\{C_Y(L_Y+L_1+L_C,t) \right\}$. 
	\item If $t_{P_i}^{S}>t_{Y_i}^{\text{max}}$, the generated pulse will reach $\max \left\{C_Y(L_Y+L_1+L_C,t) \right\}$, where the reaction between $Y$ and $P$ only influences the tail shape of the generated pulse.
\end{itemize}	
Therefore, we conclude $\max \left\{C_{\text{TX}} \right\}=\zeta C_Y(L_Y+L_1+L_C,t)$ with $\zeta \in [0,1]$. Meanwhile, the arriving time of species $P$ is determined by the length of Reaction II channel $L_2$. As such, we can flexibly control $\max \left\{C_{\text{TX}} \right\}$ by choosing different $L_2$. Based on this, we propose a step-by-step $L_2$ optimization flow as follows:
\begin{itemize}
		\item [] Initialization: Give $L_1$, $\zeta$, and initial concentrations $C_{S_{y_0}}^{\text{I}}$, $C_{X_0}^{\text{II}}$, $C_{X_0}^{\text{III}}$, and $C_{S_{p_0}}^{\text{IV}}$. 
\end{itemize}
\begin{enumerate}[Step 1:]
	\item Search for the time $t_{Y_i}^{\text{max}}$ to satisfy
	\begin{align}\label{b6}
	0\le \frac{\mathrm{d} {C_Y}(L_Y+L_1+L_C,t)}{\mathrm{d} t} \le \delta, ~t\le t_{Y_i}^{\text{max}},\\
	-\delta \le \frac{\mathrm{d} {C_Y}(L_Y+L_1+L_C,t)}{\mathrm{d} t} \le 0,~ t> t_{Y_i}^{\text{max}},
	\end{align}
	where ${C_Y}(L_Y+L_1+L_C,t)$ is given in \eqref{Y_r3}. 
	\item Calculate the maximum concentration of a generated pulse that $\max \left\{C_{\text{TX}} \right\}=\zeta {C_Y}(L_Y+L_1+L_C,t_{Y_i}^{\text{max}})$.
	\item Calculate the time $t^\text{max}_{\text{TX}_i}$ to satisfy ${C_Y}(L_Y+L_1+L_C,t^\text{max}_{\text{TX}_i})=\max \left\{C_{\text{TX}} \right\}$.  
	\item Calculate the Reaction II channel length $L_2$ via searching for
	\begin{align}
	{C_P}(L_Y+L_1+L_C,t^\text{max}_{\text{TX}_i}) \ge \epsilon, ~x\le L_2, \label{b91}\\
	{C_P}(L_Y+L_1+L_C,t^\text{max}_{\text{TX}_i}) < \epsilon, ~x>L_2, \label{b92}
	\end{align}
	where ${C_P}(L_Y+L_1+L_C,t)$ is given in \eqref{P_r3}.
\end{enumerate}

Here, we introduce two small variables, $\delta$ and $\epsilon$, to numerically find $t_{Y_i}^{\text{max}}$ and $L_2$, as it is difficult to analytically solve $\frac{\mathrm{d} {C_Y}(L_Y+L_1+L_C,t)}{\mathrm{d} t}=0$ and ${C_P}(L_Y+L_1+L_C,t^\text{max}_{\text{TX}_i})=0$.

To examine the proposed $L_2$ optimization flow, we implement three designs with different numbers of delay lines in COMSOL to achieve different $\max \left\{C_{\text{TX}} \right\}$. The implementation is shown in Fig. \ref{fig:cont} and geometric parameters are listed in Table \ref{table1} and Table \ref{table2}.
\begin{figure*}[!tb]
	\centering
	\subfloat[0 delay line, $\max \left\{C_{\text{TX}} \right\} \protect\\=\frac{1}{3} C_Y(L_Y+L_1+L_C,t)$.\label{d0}]{\includegraphics[width=1.9in]{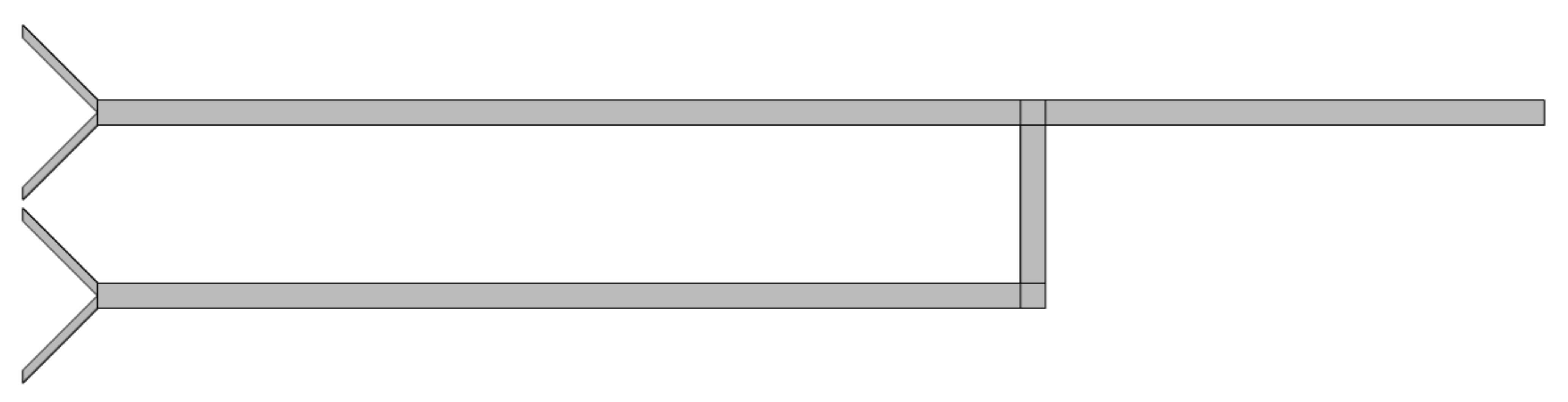}}%
	\qquad
	\subfloat[1 delay line, $\max \left\{C_{\text{TX}} \right\}\protect\\=\frac{2}{3} C_Y(L_Y+L_1+L_C,t)$.\label{d1}]{\includegraphics[width=1.9in]{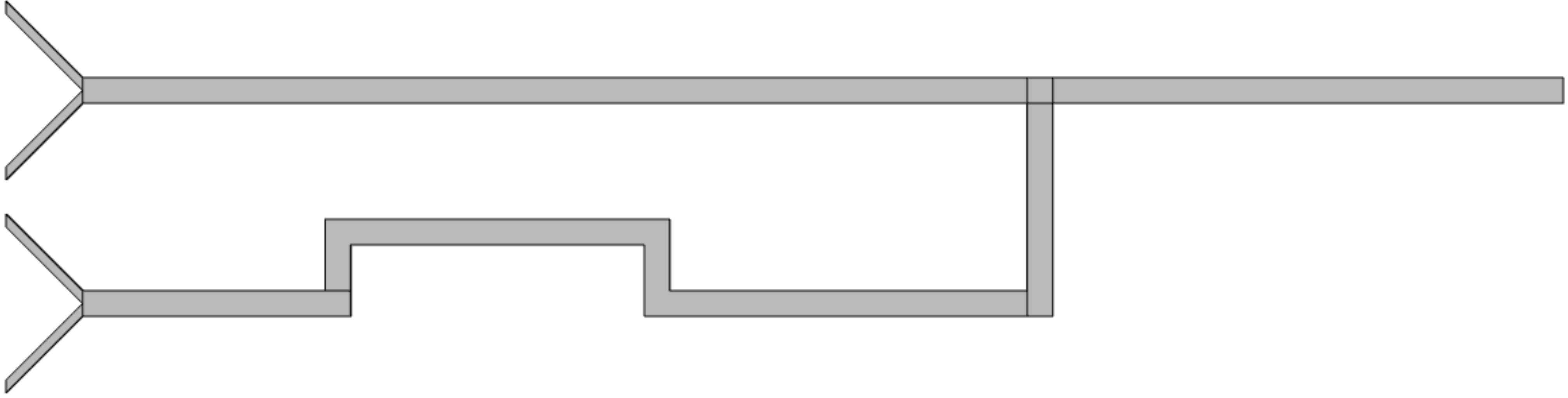}} 
	\qquad
	\subfloat[2 delay lines, $\max \left\{C_{\text{TX}} \right\}\protect\\= C_Y(L_Y+L_1+L_C,t)$.\label{d2}]{\includegraphics[width=1.9in]{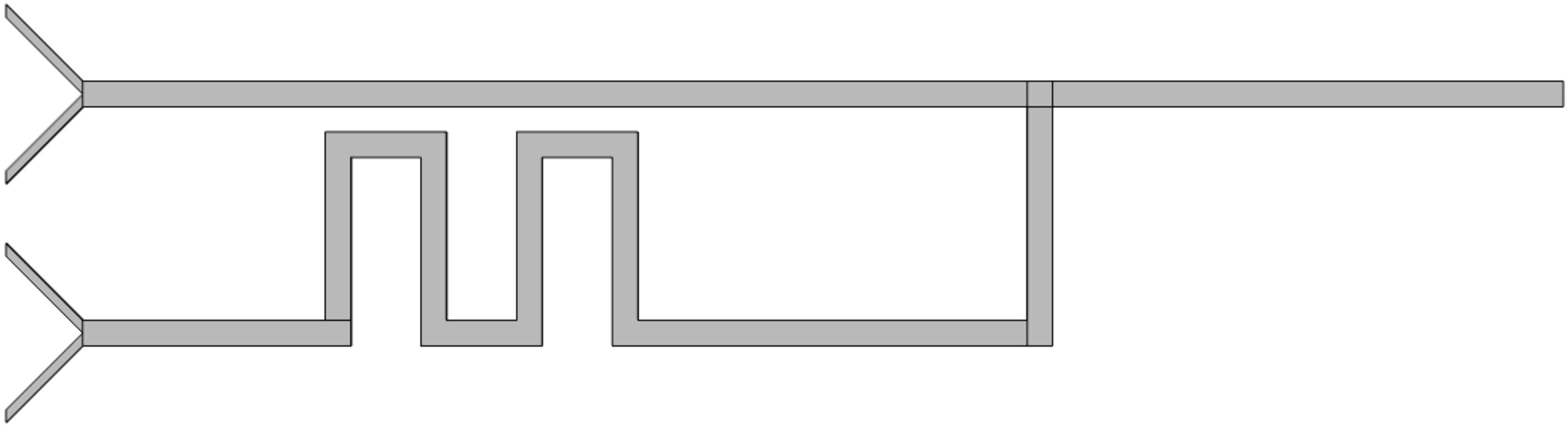}}
	\caption{Optimized transmitter implementations with different numbers of delay lines in COMSOL.}
	\label{fig:cont}%
\end{figure*}
\begin{table}[tb]
	\centering
	\caption{ The parameters of the proposed microfluidic transmitter.}
	{\renewcommand{\arraystretch}{0.8}
		\scalebox{0.8}{\begin{tabular}{c |c c c c}
				\hline
				Channel & Length($\mu$m) & Width($\mu$m)  & Depth($\mu$m)\\ \hline
				Y Junction  & $L_Y=60$ & $10$ & $10$  \\ 
				Conjunction & $L_C=20$ & $20$ & $10$  \\ 
				Reaction I Channel  & $L_1=740$ & $20$ & $10$  \\  
				Reaction III Channel  & $L_3=400$ & $20$ & $10$  \\ \hline
		\end{tabular}}
	}
	\label{table1}
\end{table}
\begin{table}[!tb]
	\centering
	\caption{The parameters of serpentine Reaction II channel in Fig. \ref{fig:cont}.}
	{\renewcommand{\arraystretch}{0.1}
		\scalebox{0.8}{\begin{tabular}{c |c c c c c c c c c c}
				\hline
				Channel & $L_2$($\mu$m) & $L_{21}$($\mu$m)  & $L_{22}$($\mu$m) &$L_{23}$($\mu$m)  &$L_{s}$($\mu$m) &$H_{s}$($\mu$m)  &$\zeta$  &$\delta$  &$\epsilon$\\ \hline
				\makecell{$0$ delay line}
				& $887$  & $/$   & $/$   &$137$ &/   &/       &$1/3$ &$0.13$&$10^{-1}$\\  \hline
				\makecell{$1$ delay line}  & $1019$ & $200$ & $300$ &$157$ &$250$ &$56$      &$2/3$ &$0.13$ &$3\times 10^{-2}$ \\ \hline
				\makecell{$2$ delay lines} & $1516$ & $200$ & $325$ &$177$ &$75$  &$147.25$  &$1$
				&$0.13$&$10^{-3}$\\  \hline
		\end{tabular}}
	}
	\label{table2} 
\end{table}
Other parameters are set following: $C_{S_{y_0}}^{\text{I}}=C_{X_0}^{\text{II}}=3$mol/m$^3$, $C_{X_0}^{\text{III}}=C_{S_{p_0}}^{\text{IV}}=4$mol/m$^3$, $D_\text{eff}=10^{-8}$m$^2$/s, $k=400$m$^3$/(mol$\cdot$s), $T_\text{ON}=2$s, $v_\text{eff}=0.2$cm/s. 
Here, we modify $\max \{C_Y(L_Y+L_1+L_C,t) \}$ from $0.75$ to $0.7498$. As shown in Fig. \ref{reaction1_channel}, when $L_1=740\mu$m, 
$C_Y(L_Y+L_1,t)$ rapidly reaches $1.4995$ at $0.55$s and then increases very slowly to the maximum concentration $1.5$ at $0.9511$s. It takes $0.4$s to reach the maximum concentration from $1.4995$, while the concentration increase is less than $0.001$. In order to generate a pulse that both two sides of the maximum concentration show a distinct increase or decrease, we modify $\max \left\{C_Y(L_Y+L_1,t) \right\}$ and $t_{Y_i}^{\text{max}}$ as $1.4995$ and $0.55$s, respectively, thus $\max \left\{C_Y(L_Y+L_1+L_C,t) \right\}=\frac{1}{2} \max \left\{C_Y(L_Y+L_1,t) \right\}=0.7498$.

In Fig. \ref{f_flow}, we plot the concentrations of generated pulses for implementations in Fig. \ref{fig:cont}.
As expected, the output pulses are generated successfully during $T_\text{ON}$, and all the maximum concentrations of the pulses reach their corresponding analytical values (marked in black dash-dot lines). It is also seen that the longer the Reaction II channel is, the wider the generated pulse, because of the longer time given to reach a higher maximum concentration. These observations reveal the dependency of the maximum concentration of a generated pulse on the Reaction II channel length $L_2$, show how the predefined shaping of the pulse can be controlled, and highlight the importance of deriving theoretical signal responses in design stage.

\subsubsection{\textbf{Optimal Design of the Restricted Time Gap}}
\begin{figure}[t]
	\centering
	\begin{minipage}[t]{0.47\textwidth}
		\centering
		\includegraphics[width=3.4in]{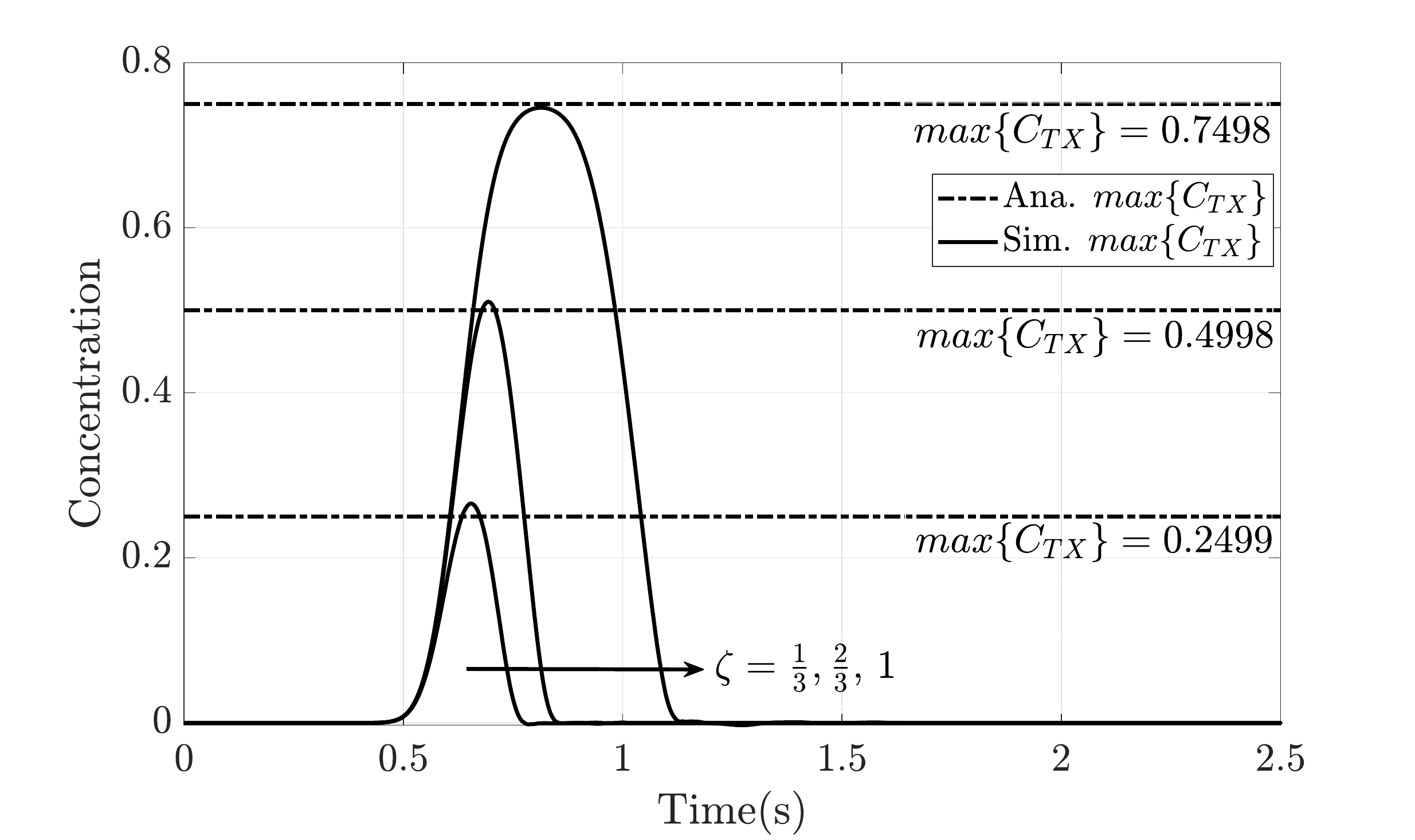}
		\caption{The concentrations of generated pulses for different transmitter implementations.}
		\label{f_flow}
	\end{minipage}
	\quad
	\begin{minipage}[t]{0.47\textwidth}
		\centering
		\includegraphics[width=2.6in]{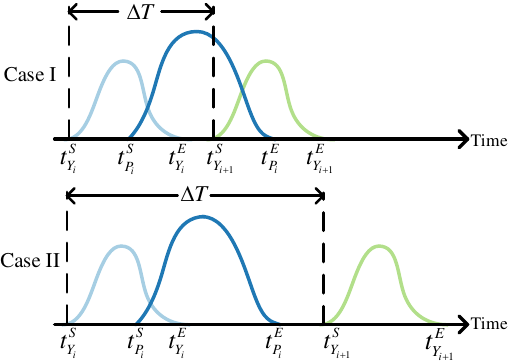}
		\caption{The concentrations of species $Y$ and $P$ at Reaction III channel inlet with different time gaps.}
		\label{deltaT}
	\end{minipage}
\end{figure}

The design that the Reaction II channel is longer than the Reaction I channel ($L_2>L_1$) is also likely to cause distorted pulses if the time gap $\Delta T$ between two consecutive input bits is not chosen appropriately. Assuming that species $Y$ generated by the $(i+1)$th input bit arrives earlier than the leaving time of species $P$ generated by the $i$th input bit at Reaction III channel inlet, $Y$ will be immediately consumed according to $\text{Reaction III}$ when they simultaneously enter the Reaction III channel so that the maximum concentration of the generated pulse for the $(i+1)$th input bit is distorted and less than $\max \left\{C_{\text{TX}} \right\}$. To prevent this, the time gap $\Delta T$ should be restricted. 

Remind that the arriving and leaving time of a general species $A$ at Reaction III channel inlet are denoted as $t_{A_i}^{S}$ and $t_{A_i}^{E}$ for the $i$th input bit.
As shown in Fig. \ref{deltaT}, species $Y$ generated by the $(i+1)$th input bit can appear earlier in Case I or later in Case II than species $P$ generated by the $i$th input bit via adjusting $\Delta T$. In Case I, the earlier arriving of $Y$ makes itself react with the tail of $P$, thus breaking the principle that $Y$ should increase to $\max \left\{C_{\text{TX}} \right\}$ and then drop to zero. To avoid this, $\Delta T$ needs to satisfy
\begin{align}
	\Delta T \ge t_{P_i}^{E}-t_{Y_i}^{S},
\end{align}
where $t_{Y_i}^{S}$ and $t_{P_i}^{E}$ can be numerically solved by
\begin{align} 
C_Y(L_Y+L_1+L_C,t) \leq \tau, ~t\leq t_{Y_i}^{S},~C_Y(L_Y+L_1+L_C,t) > \tau,~ t> t_{Y_i}^{S}; \label{searching_delta1}\\
C_P(L_Y+L_1+L_C,t) \ge \tau, ~t\leq t_{P_i}^{E}, ~C_P(L_Y+L_1+L_C,t) < \tau, ~t> t_{P_i}^{E} \label{searching_delta2}.
\end{align}
Here, $\tau$ is a small variable to find $t_{Y_i}^{S}$ and $t_{P_i}^{E}$ that $C_Y(L_Y+L_1+L_C,t_{Y_i}^{S})=0$ and $C_P(L_Y+L_1+L_C,t_{P_i}^{E})=0$, respectively.

In Fig. \ref{f_deltaT}, we plot the concentrations of species $Y$ and $P$ at Reaction III channel inlet and the generated pulses with different $\Delta T$.
\begin{figure*}%
	\centering
	\subfloat[The durations of two consecutive input bits are ${[0.1,2.1]}$ and ${[2.4,4.4]}$\label{f_deltaT1}]{\includegraphics[width=3in]{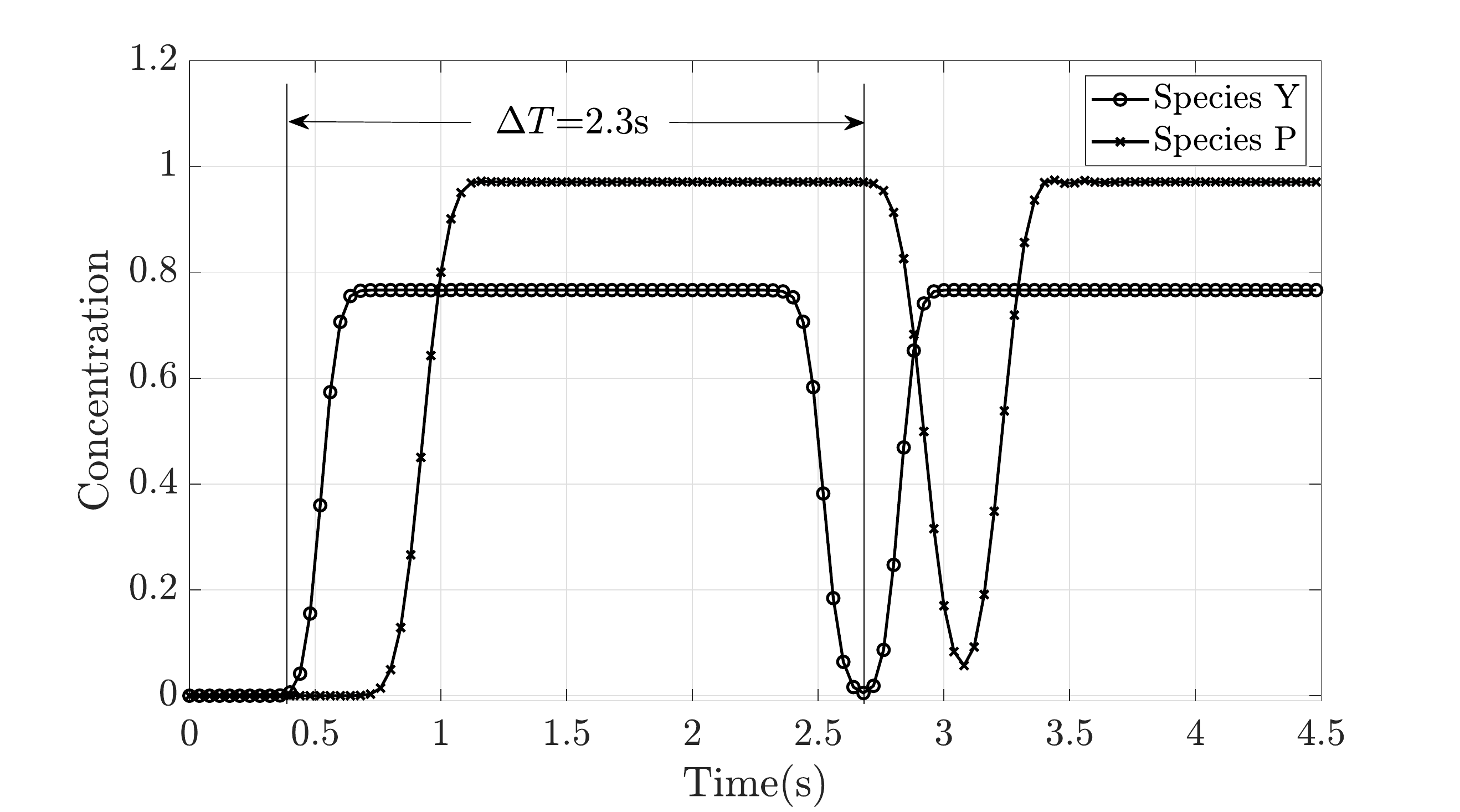}}%
	\qquad
	\subfloat[The durations of two consecutive input bits are ${[0.1,2.1]}$ and ${[3.1,5.1]}$.\label{f_deltaT5}]{\includegraphics[width=3in]{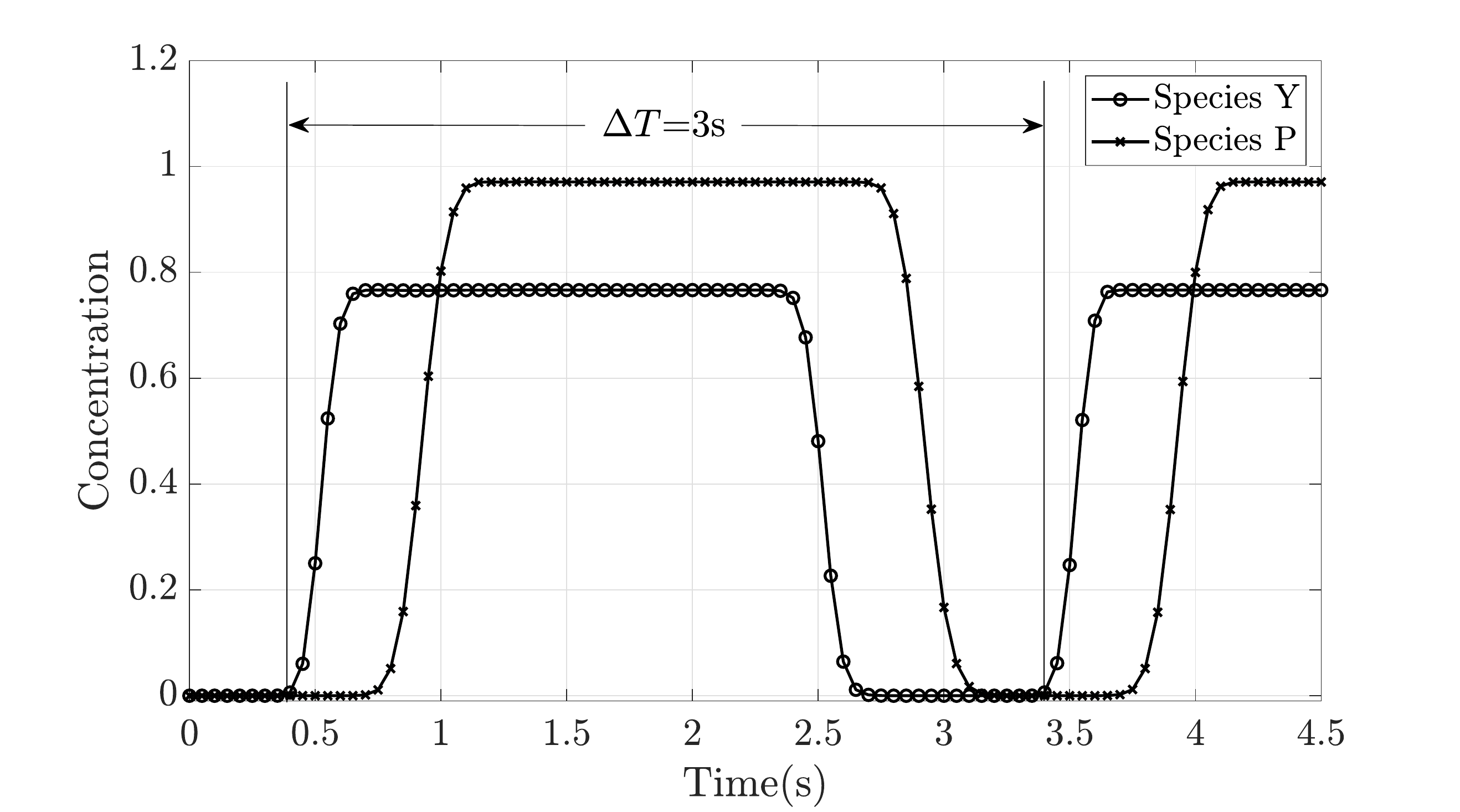}}\\
	\subfloat[The generated pulses of (a).\label{f_deltaT2}]{\includegraphics[width=3in]{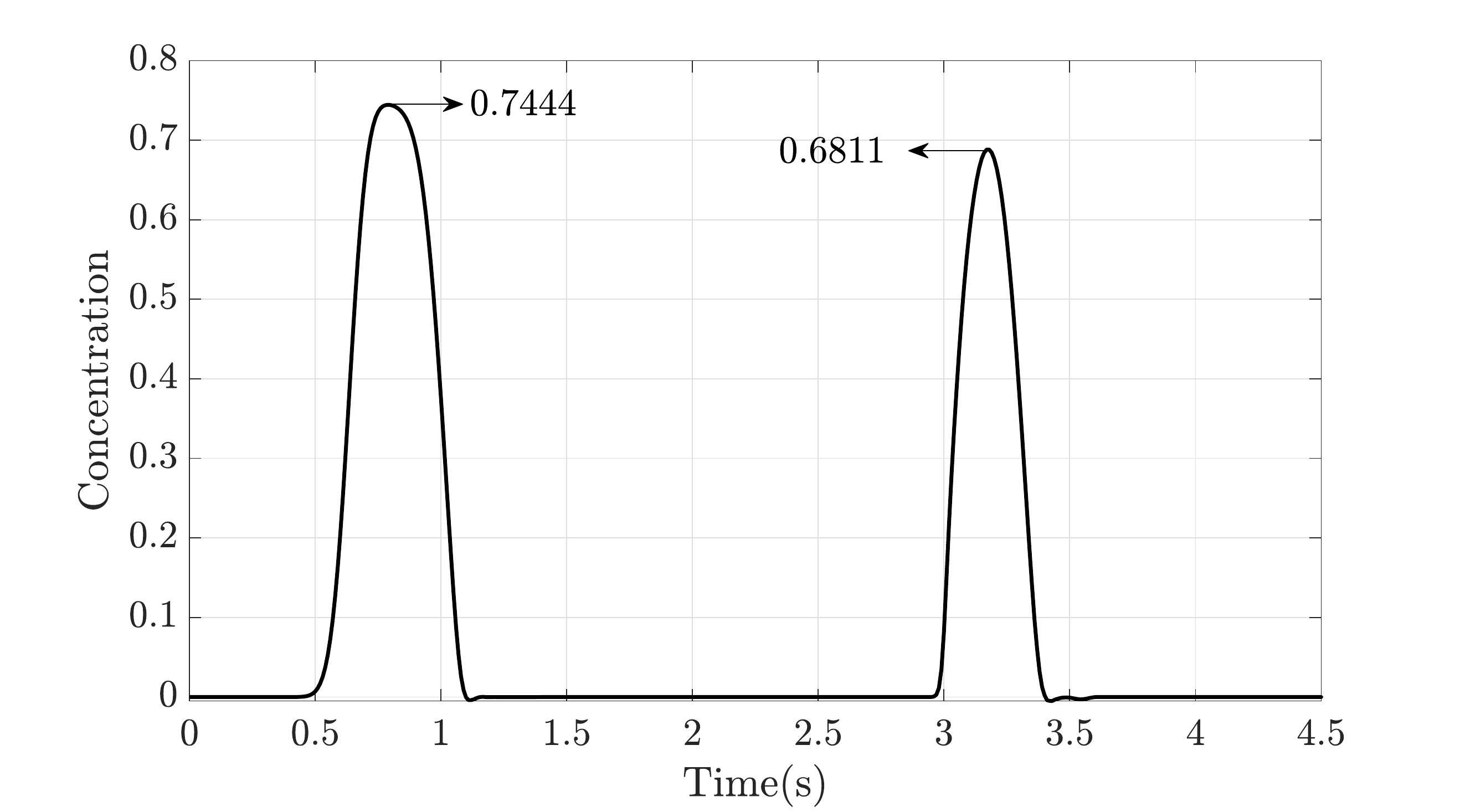}} 
	\qquad
	\subfloat[The generated pulses of (b).\label{f_deltaT6}]{\includegraphics[width=3in]{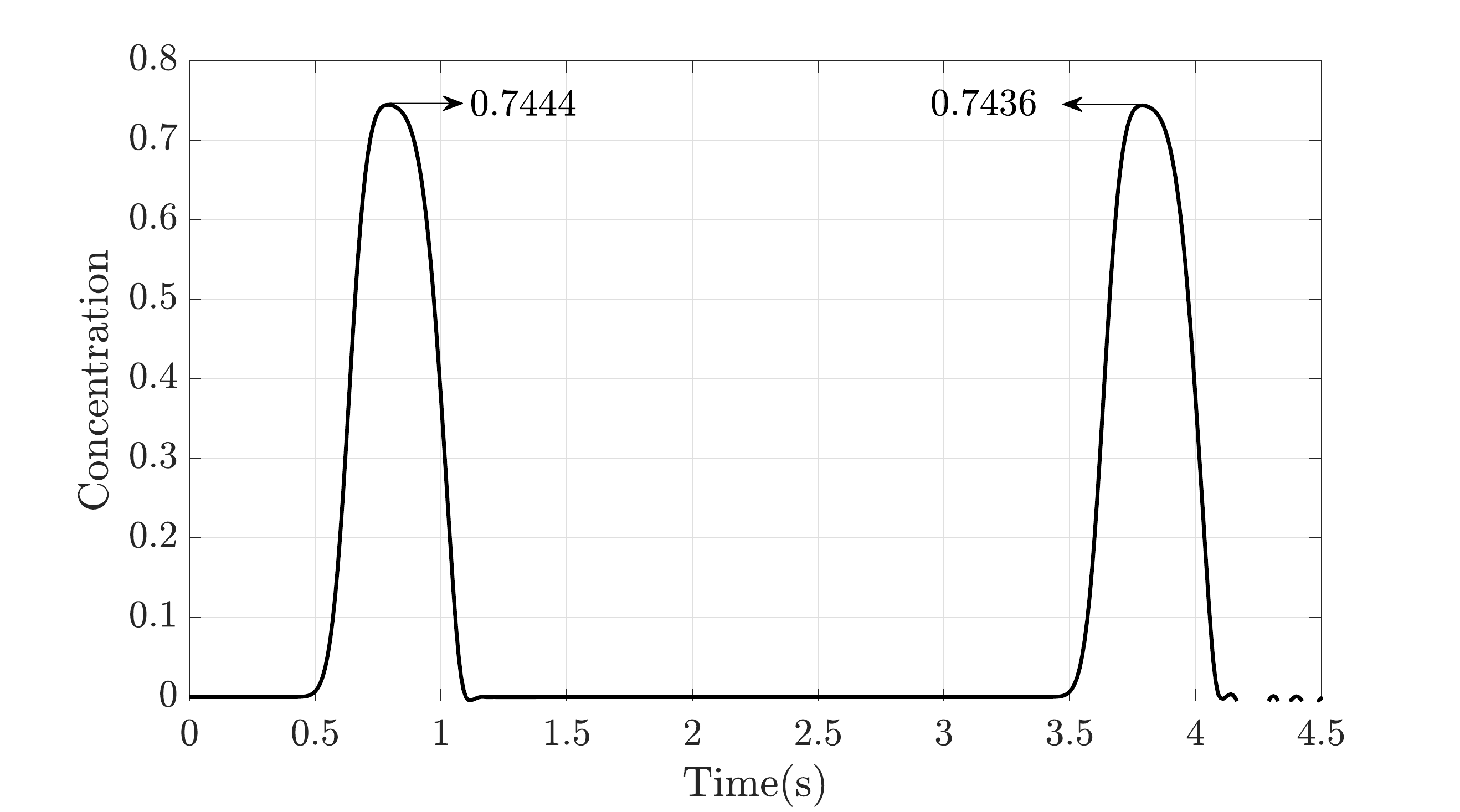}}
	\caption[The study of the $\Delta T$.]{The concentrations of species $Y$ and $P$ at Reaction III channel inlet and their generated pulses with different $\Delta T$.}
	\label{f_deltaT}%
\end{figure*}
We use the parameters for Fig. \ref{fig:cont} \subref{d2} and $\tau=10^{-3}$. We numerically solve \eqref{searching_delta1}, \eqref{searching_delta2} and obtain $\Delta T \ge 2.75$s. Fig. \ref{f_deltaT} \subref{f_deltaT2} shows that the second pulse is distorted compared with the first pulse because $\Delta T=2.3$s leading to 
the earlier arriving of species $Y$ generated by the $2$nd input bit, and thus a twice consumption of $Y$, being first consumed by the tail of $P$ generated by the $1$st input bit and then by the arriving of $P$ generated by the $2$nd input bit. On the contrary, Fig. \ref{f_deltaT} \subref{f_deltaT6} illustrates a generation of two non-distorted and identical-shaped pulses with a satisfied $\Delta T$.

\vspace{-8pt}	
\section{Microfluidic MC Receiver Analysis and Design}
\label{sec:RX}
In this section, we analyse the T Junction and two reaction channels, and then provide some guidelines on how to design a microfluidic MC receiver.
\vspace{-10pt}	
\subsection{Microfluidic MC Receiver Analysis}
\subsubsection{\textbf{T Junction}}
After information propagation, the transmitted molecules $Y$ from microfluidic transmitter propagates to enter the receiver through Inlet V. Here, we set the location of Inlet V as the position origin and the time that species $Y$ flows into Inlet V as the time origin. Since the transmitted pulse cannot be theoretically characterized, let us assume that a received pulse follows a Gaussian concentration distribution with mean $\mu$ and variance $\sigma^2$, that is
\begin{align} \label{er1}
C_Y(0,t)=\frac{C_{Y_0}^{\text{V}}}{\sqrt{2 \pi \sigma^2}} e^{-\frac{(t-\mu)^2}{2\sigma^2}}.
\end{align}
As the length of one T junction branch $L_T$ is much shorter than that of the following reaction channel, and no reaction happens in a T junction, we further assume the concentration of species $Y$ at T junction I outlet as 
\begin{align} \label{assumption_r1}
C_Y(L_T+L_C,t)\approx \frac{1}{2}C_Y(0,t-t_\text{T}),
\end{align}
where the $\frac{1}{2}$ describes the dilution of species $Y$ by species $ThL$ who is continuously injected into Inlet VI with concentration $C_{ThL}^{\text{VI}}$, and $t_\text{T}=\frac{L_T+L_C}{v_\text{eff}}$ is the travelling time over T junction I via average velocity $v_{\text{eff}}$. Similarly, the outlet concentration of species $ThL$ is assumed as
\begin{align} \label{assumption_r2}
C_{ThL}(L_T+L_C,t)\approx \frac{1}{2}C_{ThL}^{\text{VI}}, ~t\geq t_\text{T}.
\end{align} 
\subsubsection{\textbf{Straight Reaction IV Channel}}
The outflow of T junction I flows through the Reaction IV channel with length $L_4$ to proceed Reaction IV (the thresholding reaction) in \eqref{CR4}, where the portion of species $Y$, whose concentration below $\frac{1}{2}C_{ThL}^{\text{IV}}$, is depleted by reactant $ThL$. With assumptions of \eqref{assumption_r1} and \eqref{assumption_r2}, the concentration of species $Y$ at Reaction IV channel outlet can be expressed using \eqref{approx1} or \eqref{approx2} in \textbf{Theorem 2} by substituting $C_{A_0}$ and $C_{B_0}$ with $C_{Y_0}^{\text{V}}$ and $C_{ThL}^{\text{VI}}$, which yields
\begin{align} \label{appro_r4}
	C_Y(L_T+L_C+L_4,t) \approx   \frac{1}{2}C_{A}^{\text{Appro$_1$}}(L_4,t-t_\text{T})~~\text{or}~~\frac{1}{2}C_{A}^{\text{Appro$_2$}}(L_4,t-t_\text{T}) .
\end{align}

\subsubsection{\textbf{Straight Reaction V Channel}}
After $\text{Reaction IV}$, the remaining species $Y$ flows into the Reaction V channel and catalyses the conversion of species $Amp$ into output species $O$, where $Amp$ is continuously infused with constant concentration $C_{Amp}^{\text{VII}}$ into Inlet VII. As a catalyst, species $Y$ does not react with species $Amp$, and the produced quantity of species $O$ equals the reacting concentration of $Amp$ according to their stoichiometric relation. Considering the dilution at T junction II, the reacting concentration of $Amp$ is diluted to one third of its injected concentration by flows injected at Inlet V and Inlet VI. Based on this and ignoring the diffusion effect in Reaction V channel, the demodulated signal containing species $O$ can be approximated as
\begin{align}\label{eo}
C_{O}(t)=\begin{cases}
\frac{1}{3}C_{Amp}^{\text{VII}}, ~~~&C_{Y}(L_T+2L_C+L_4+L_5,t-\frac{L_C+L_5}{v_\text{eff}})\ge 0 \\
0, &\text{otherwise}.
\end{cases}
\end{align}

\subsubsection{\textbf{Simulation Results}}
To examine the microfluidic receiver analysis, we implement the receiver design in COMSOL (shown in Fig. \ref{frx_design}) with geometric parameters listed in Table \ref{table3}.
\begin{figtab}[t]
	\centering
	\begin{minipage}[b]{0.37\textwidth}
	\centering
	\includegraphics[width=2.5in]{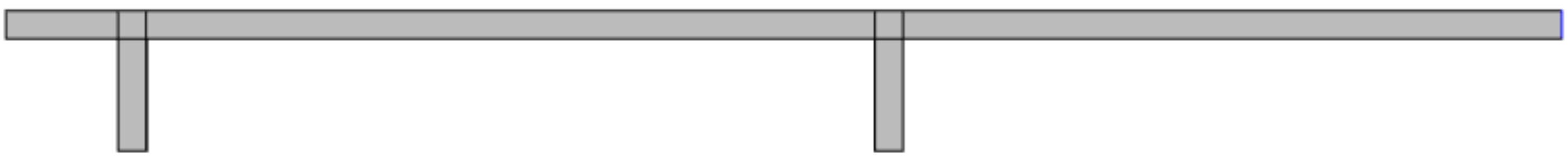}
	\figcaption{The proposed microfluidic receiver implementation in COMSOL.}
	\label{frx_design}
	\end{minipage}
	\quad
	\begin{minipage}[b]{0.58\textwidth}
				\centering
				\tabcaption{The parameters of the proposed microfluidic receiver.}
				\resizebox{190pt}{30pt}{\begin{tabular}{c |c c c c}
						\hline
						Channel & Length($\mu$m) & Width($\mu$m)  & Depth($\mu$m)\\ \hline
						T Junction  & $L_T=80$ & $20$ & $10$  \\ 
						Conjunction  &$L_C=20$ & $20$ & $10$  \\ 
						Reaction IV Channel & $L_4=520$ & $20$ & $10$  \\ 
						Reaction V Channel & $L_5=470$ & $20$ & $10$  \\ \hline
				\end{tabular}}	
			\label{table3}	
	\end{minipage}
\end{figtab}
We set the parameters: $C_{Y_0}^{\text{V}}=3$mol/m$^3$, $\mu=2$, $\sigma^2=0.25$, $k=400$m$^3$/(mol$\cdot$s), $D_\text{eff}=10^{-8}$m$^2$/s, and $v_\text{eff}=0.2$cm/s.

Fig. \ref{f_r1} compares the concentration of species $Y$ at Reaction IV channel outlet with the two approximations in \eqref{appro_r4}. We observe that the two approximations have short delay compared with simulation results due to the propagation during T junction I. The reason for this difference is that even though the T Junction length is shorter than the Reaction IV channel length, the concentration of species $Y$ at T Junction I outlet is not merely a time shift of the received Gaussian concentration in \eqref{er1} as the diffusion effect can broaden the Gaussian concentration through travelling T Junction I.

Fig. \ref{frx3} demonstrates the significant role of $C_{ThL}^{\text{VI}}$ on the width of the demodulated signal $C_O(t)$. As $C_{ThL}^{\text{VI}}$ increases, the width of the demodulated signal decreases. If $C_{ThL}^{\text{VI}} >\max \left\{C_{Y}(0,t) \right\}$, we expect that there is no residual $Y$ in Reaction V channel, so that species $O$ cannot be produced. Fig. \ref{frx4} plots the concentrations of species $O$ at Reaction V channel outlet with different $C_{Amp}^{\text{VII}}$. As expected, the outlet concentration of species $O$ varies with $C_{Amp}^{\text{VII}}$, and approximately equals $\frac{1}{3}C_{Amp}^{\text{VII}}$, which reveals that it is possible to reach any level $C_O$ via adjusting $C_{Amp}^{\text{VII}}$. 

\begin{figure}[t]
	\centering
	\includegraphics[width=3.2in]{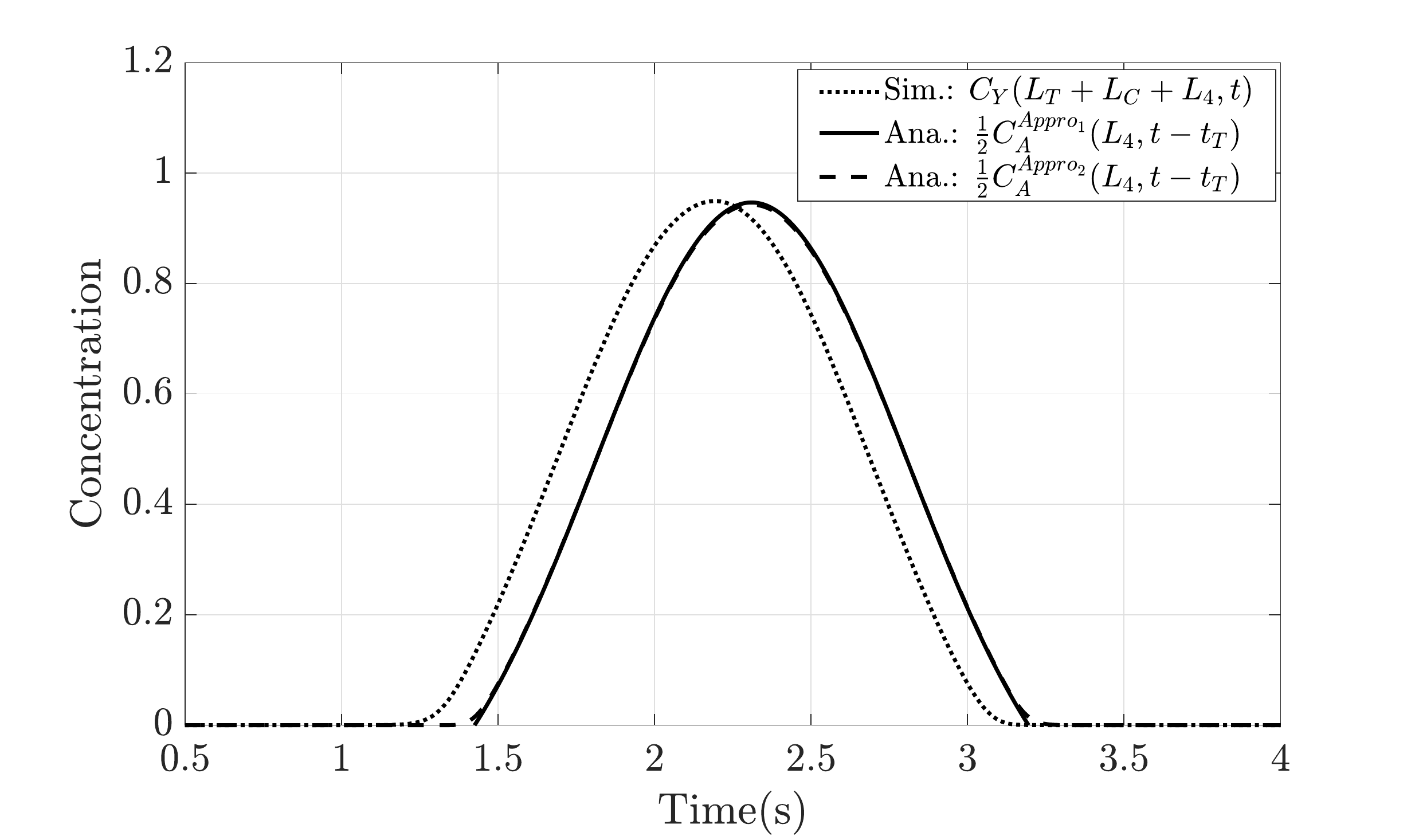}
	\caption{The concentration of species $Y$ at Reaction IV channel outlet with T Junction I.}
	\label{f_r1}
\end{figure}
\begin{figure}[t]
	\centering
	\begin{minipage}[t]{0.47\textwidth}
		\centering
		\includegraphics[width=3in]{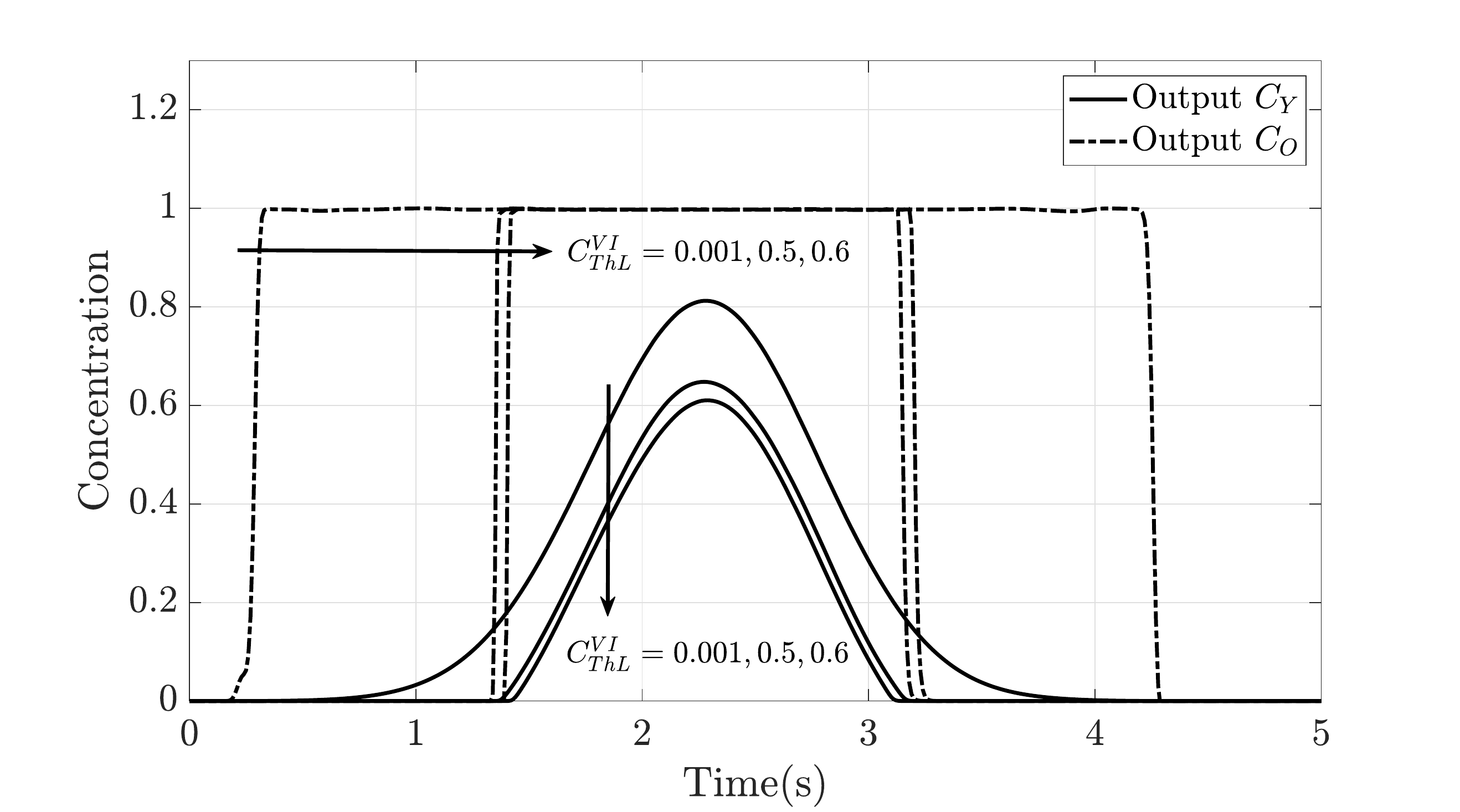}
		\caption{The concentrations of species $Y$ and $O$ at Reaction V channel outlet with different $C_{ThL}^{\text{VI}}$, where the concentration of species $O$ is normalized to $1$mol/m$^3$.}
		\label{frx3}
	\end{minipage}
	\quad
	\begin{minipage}[t]{0.47\textwidth}
		\centering
		\includegraphics[width=3in]{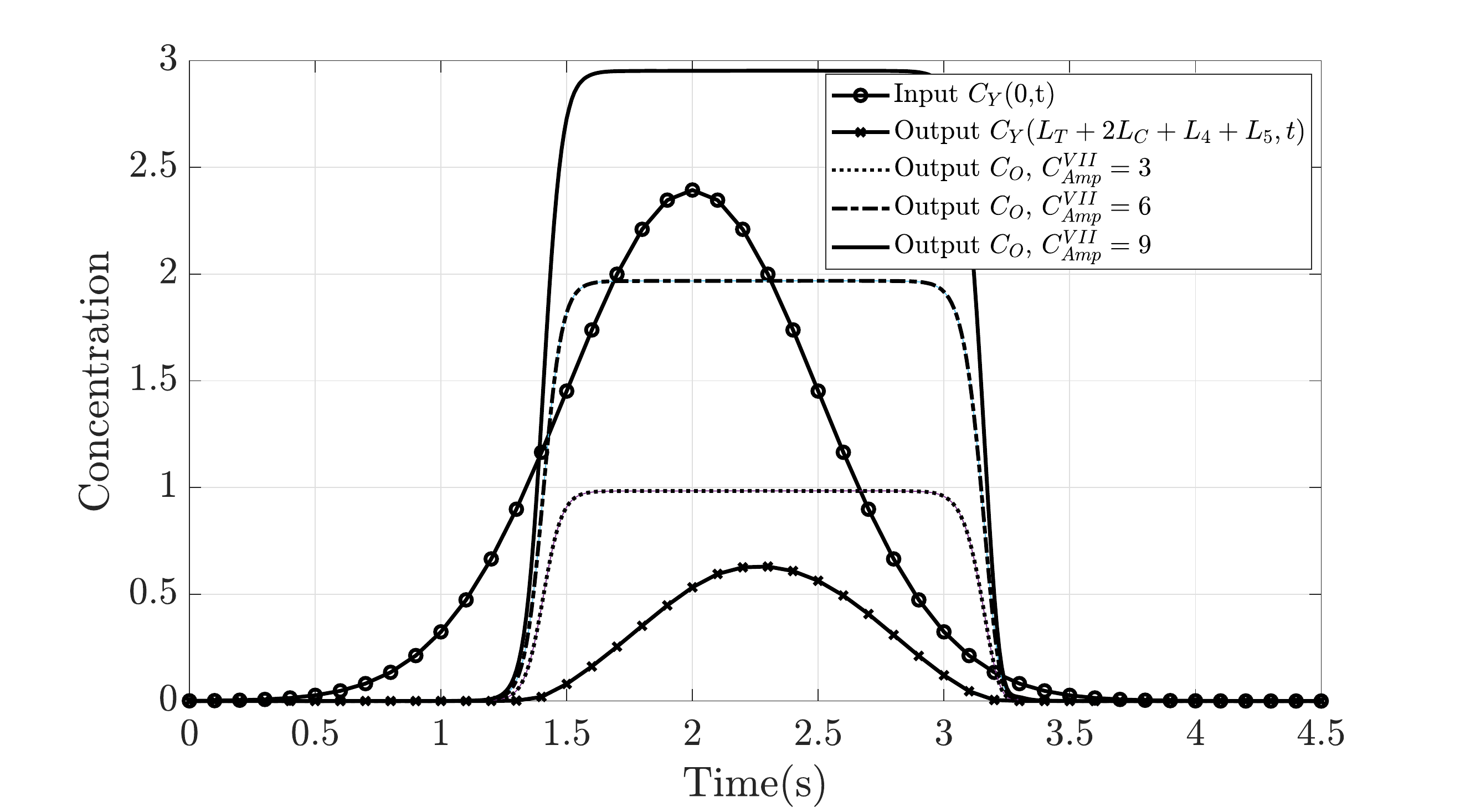}
		\caption{The outlet concentrations of species $O$ at Reaction V channel with different $C_{Amp}^{\text{VII}}$.}
		\label{frx4}
	\end{minipage}
\end{figure}

\vspace{-12pt}
\subsection{Microfluidic MC Receiver Design}
Based on the simulation results in Fig. \ref{frx3} and \ref{frx4}, we conclude two receiver design guidelines. First, the results in Fig. \ref{frx3} reveal that the demodulated signal width is dependent on $C_{ThL}^{\text{VI}}$, and $C_{ThL}^{\text{VI}}$ cannot exceed the maximum concentration of a received pulse, which in turn highlights the necessity and importance to study the maximum concentration control of a generated pulse in Sec. \ref{optimization1}. Second, the results in Fig. \ref{frx4} present the relation between $C_{Amp}^{\text{VII}}$ and $C_O$ follows $C_O=\frac{1}{3}C_{Amp}^{\text{VII}}$. This insight is helpful in concentration detection. If concentration is detected through fluorescence, the relation $C_O=\frac{1}{3}C_{Amp}^{\text{VII}}$ enables us to determine how much $C_{Amp}^{\text{VII}}$ should be injected to ensure fluorescent species $O$ to be captured by a microscopy.

\vspace{-12pt}	
\section{An End-to-End Microfluidic MC Implementation}
\label{MC}
In this section, we combine the microfluidic transmitter with the receiver as proposed in Fig. \ref{bf2} to form a basic end-to-end MC system, where the transmitter and the receiver share the same design parameters as implementations in Fig. \ref{f_deltaT} \subref{f_deltaT5} and Fig. \ref{frx_design}, and the propagation channel is a straight convection-diffusion channel with length $1000\mu$m. 
Considering the reacting concentration of species $Amp$ is diluted to one fourth of its injected concentration $C_{Amp}^{\text{VII}}$ by flows from Y Junction I outlet, Y Junction II outlet, and Inlet VI, we set $C_{Amp}^{\text{VII}}=12$mol/m$^3$ for the purpose of restoring the output concentration level to input concentration of species $X$ injected at Inlet II ($C_{X_0}^{\text{II}}=3$mol/m$^3$).

\begin{figure}[t]
	\centering
	\includegraphics[width=3.2in]{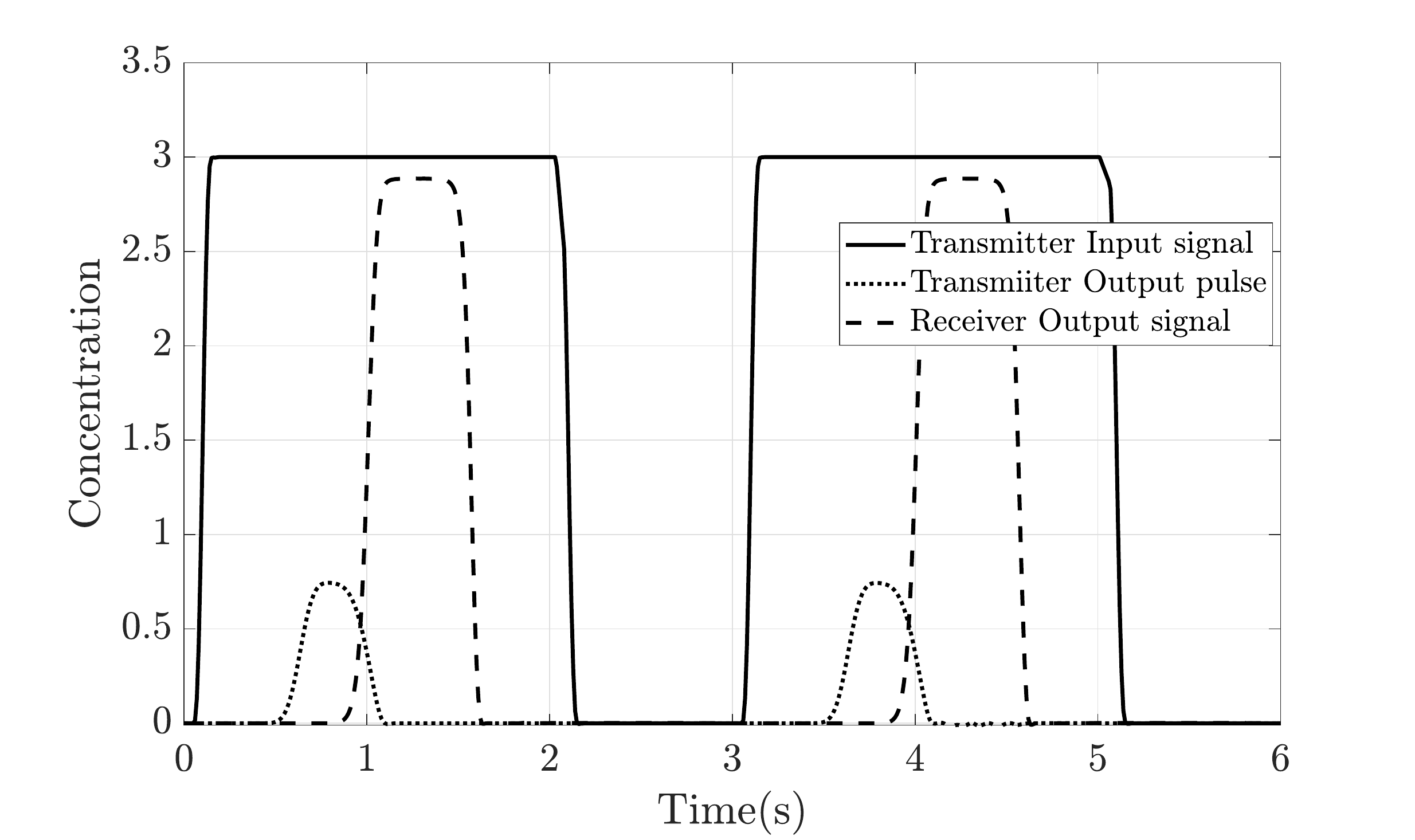}
	\caption{The transmitter input signals, transmitter output pulses, and receiver output signals for the basic end-to-end MC implementation.}
	\label{fsp}
\end{figure} 
Fig. \ref{fsp} plots the transmitter input signals, transmitter output pulses, and receiver output signals. It is clear that two consecutive rectangular signals are successfully demodulated, and this result demonstrates the validity of the end-to-end MC system. Moreover, we observe that although the concentrations of transmitter output pulses are much lower than concentrations of transmitter input signals due to two dilutions occurred on Y Junction output and the conjunction between Reaction I/II channel and Reaction III channel, the concentrations of receiver output signals can approximately restore to the same concentration level of input signal via adjusting $C_{Amp}^{\text{VII}}$.

\vspace{-12pt}	
\section{Conclusion}
\label{sec:conclusion}
In this paper, we expanded our previous work on the design and analysis of a microfluidic MC transmitter capable of generating a pulse-shaped molecular concentration upon a rectangular triggering input with chemical reactions. We further proposed a microfluidic receiver design based on a thresholding reaction and an amplifying reaction to realize a function of demodulating a received signal into a rectangular output signal. Both the proposed designs were based on microfluidic systems with standard and reproducible components, and these microfluidic components were analytically characterized to reveal the dependence of the generated pulse and the demodulated signal on design parameters. For transmitter design optimization, we proposed a reaction channel length optimization flow to control the maximum concentration of output pulse at the transmitter, and then derived a time gap constraint between two consecutive input bits to ensure a continuous transmission of non-distorted and identical-shaped pulses. Finally, we implemented an end-to-end microfluidic MC system through connecting the transmitter with the receiver, and simulation results performed in COMSOL Multiphysics demonstrated successfully pulse generation and signal demodulation, thus effectiveness of the proposed designs. Notably, our proposed microfluidic transceiver will act as fundamental building blocks in the design of future micro/nanoscale MC systems. More importantly, the methodology presented in this paper will inspire the design of additional MC blocks inspired by biochemical processes and based on microfluidic systems.

%

\vspace{-12pt}	
\appendices
\section{Proof of theorem 1}
\label{A}
To solve the spatial-temporal concentration distributions of species $A$ and $AB$, we first define some initial boundary conditions. Species $A$ and $B$ are injected at the inlet of a straight microfluidic channel $x=0$, the first initial boundary condition is
\begin{equation}
\begin{aligned}
C_A(0,t) &= \min\left\{C_{A_0}, C_{B_0} \right\}=C_{0},  ~0\le t \le  T_\text{ON} \\
&= C_0[u(t)-u(t-T_\text{ON})]. \label{reactionb1}
\end{aligned}
\end{equation} 
Here, 
it must be careful that $C_A(0,t)$ may not equal its injected concentration. This is because the one-to-one stoichiometric relation between species $A$ and $B$ in $A+B\to AB$ determines that either the reacting concentration of species $A$ or $B$ equals the smaller supplied concentration, \textit{i.e.}, $C_A(0,t) = \min\left\{C_{A}(0,t), C_{B_0} \right\}$. At $t=0$, the concentration of species $A$ in any positions is zero, thus the second initial boundary condition being 
\begin{align}
C_A(x,0) = 0,  ~x \ge  0. \label{reactionb2}
\end{align}
As the concentration change over locations far away from the source equals zero so that the third boundary condition is
\begin{align}
\frac{{\partial C_A(\infty ,t)}}{{\partial x}} = 0,  ~t \ge  0. \label{reactionb3}
\end{align}

The concentration distribution can be obtained by taking the Laplace transform of \eqref{reaction1}, \eqref{reactionb1}, and \eqref{reactionb3}  using
\begin{align}
\widetilde{C_A}\left( {x,s} \right) = \int_0^\infty  {{e^{ - st}}C_A^{\text{}} \left( {x,t} \right)\mathrm{d}t}.  \label{reactionb4}
\end{align}
The Laplace transform of  \eqref{reaction1} satisfying   \eqref{reactionb2} is
\begin{align}
{D_\text{eff}}\frac{{{\partial ^2}\widetilde{{C_A^{\text{}}}}(x,s)}}{{\partial {x^2}}} - {v_\text{eff}}\frac{{\partial \widetilde{{C_A^{\text{}}}}(x,s)}}{{\partial x}} = \left( {s + {k}{C_{{{0}}}}} \right)\widetilde{{C_A^{\text{}}}}(x,s). \label{reactionb5}
\end{align}
The Laplace transforms of \eqref{reactionb1} and  \eqref{reactionb3} can be expressed 
\begin{align}
\widetilde{{C_A^{\text{}}}}(0 ,s) = \frac{C_0}{s}(1-e^{-T_\text{ON}s}), \label{reactionb55}\\
\text{and} \;\;\frac{{\partial \widetilde{{C_A^{\text{}}}}(\infty ,s)}}{{\partial x}} = 0.\label{reactionb6}
\end{align}
Combining  \eqref{reactionb5},   \eqref{reactionb55}, and  \eqref{reactionb6}, we derive
\begin{align}
\widetilde{{C_A^{\text{}}}}(x,s) = \frac{{C_0}(1-e^{-T_\text{ON}s})}{s}\exp \left[ {\frac{{{v_\text{eff}}x}}{{2{D_\text{eff}}}} - x\sqrt {\frac{{{{v_\text{eff}}^2}}}{{4{{D_\text{eff}}^2}}} + \frac{{s + {k}{C_{{0{}}}^{\text{}}}}}{D_\text{eff}}} } \right]. \label{reactionb7}
\end{align}
Taking the inverse Laplace transform of \eqref{reactionb7}, we derive
\begin{equation}\label{b4_proof}
{C_A^{\text{}}}(x,t)=\begin{cases}
g(x,t), &0\leq t \leq T_\text{ON}\\
g(x,t)-g(x,t-T_\text{ON}), &t> T_\text{ON},
\end{cases}
\end{equation}
where
\begin{equation}
\nonumber
	g(x,t) =\frac{C_0}{2} \left\{ {\exp \left[ {\frac{{\left( {{v_\text{eff}} - \alpha } \right)x}}{{2{D_\text{eff}}}}} \right]\text{erfc}\left[ {\frac{{x - \alpha t}}{{2\sqrt {{D_\text{eff}}t} }}} \right]} \right. \left. { + \exp \left[ {\frac{{\left( {{v_\text{eff}} + \alpha } \right)x}}{{2{D_\text{eff}}}}} \right]\text{erfc}\left[ {\frac{{x + \alpha t}}{{2\sqrt {{D_\text{eff}}t} }}} \right]} \right\}
\end{equation}
with $\alpha  = \sqrt {{{v_\text{eff}}^2} + 4{k}{C_{{0{}}}}{D_\text{eff}}}  $.

To derive the concentration of species $AB$, we combine \eqref{reaction1} and \eqref{reaction1Y} as
\begin{align}
&{D_\text{eff}}\frac{{{\partial ^2}{C_s}(x,t)}}{{\partial {x^2}}} - {v_\text{eff}}\frac{{\partial {C_s}(x,t)}}{{\partial x}} = \frac{{\partial {C_s}(x,t)}}{{\partial t}} , \label{reaction1com}
\end{align}
where ${C_s}(x,t)=  {C_A^{\text{}}}(x,t)+{C_{AB}^{\text{}}}(x,t)$. Interestingly, this equation is the convection-diffusion equation for the total concentration distribution of molecule $A$ and $AB$.
The sum concentration of $A$ and $AB$ follows the three boundary conditions 
\begin{align}
C_s(0,t) = C_0,  ~0\le t \le  T_\text{ON} ,\\
C_s(x,0) = 0,  ~x \ge  0, \\
\text{and} \;\; C_s(\infty ,t) = 0,~  t \ge  0 .
\end{align}
Following \cite[eq. (11)]{Yansha16}, we can derive the molecular concentration as
\begin{equation}\label{b3}
C_s(x,t)=\begin{cases}
h(x,t), &0\leq t \leq T_\text{ON}\\
h(x,t)-h(x,t-T_\text{ON}), &t> T_\text{ON},
\end{cases}
\end{equation}
where $h(x,t)=\frac{{C_{0}}}{2}\left[ {\text{erfc}\left( {\frac{{x - {v_\text{eff}}t}}{{2\sqrt {{D_\text{eff}}t} }}} \right) + {e^{\frac{{{v_\text{eff}}x}}{{D_\text{eff}}}}}\text{erfc}\left( {\frac{{x + {v_\text{eff}}t}}{{2\sqrt {{D_\text{eff}}t} }}} \right)} \right]$. Taking the deduction of ${C_A^{\text{}}}(x,t)$ in  \eqref{b4_proof} from $C_s\left( {x,t} \right)$, we derive the concentration of $AB$ as
\begin{equation}\label{b5_proof}
{C_{AB}}(x,t)=\begin{cases}
h(x,t)-g(x,t), &0\leq t \leq T_\text{ON}\\
[h(x,t)-g(x,t)]-[h(x,t-T_\text{ON})-g(x,t-T_\text{ON})], &t> T_\text{ON}.
\end{cases}
\end{equation}

\section{Proof of theorem 2}
\label{B}
Similar to the Proof of \textbf{Theorem 1}, we first define initial boundary conditions. On the condition of $C_{B_0}>\max\left\{ C_A(0,t) \right\}$ and due to the one-to-one stoichiometric relation between $A$ and $B$, the first initial condition varies with $C_{B_0}$, and can be expressed as
\begin{equation}\label{er3}
C_A(0,t)=\begin{cases}
C_A(0,t), &0\leq t < t_1\\
C_{B_0}, &t_1\leq t <t_2\\
C_A(0,t), &t_2\leq t ,
\end{cases}
\end{equation}
where $t_1$ and $t_2$ are obtained through solving $C_A(0,t)=C_{B_0}$, and finally $t_1 =\mu -\sqrt{-2\sigma^2 \ln \frac{C_{B_0}\sqrt{2\pi \sigma^2}}{C_{A_0}^{\text{}}}}$ and $t_2 =\mu +\sqrt{-2\sigma^2 \ln \frac{C_{B_0}\sqrt{2\pi \sigma^2}}{C_{A_0}^{\text{}}}}$. 
The second and third initial boundary conditions are the same with \eqref{reactionb2} and \eqref{reactionb3}, respectively. Next, we introduce two approximation methods to solve (\ref{reaction1}), where we split the fully coupled convection-diffusion-reaction process into two sequential processes: 1) the reaction process (described by a reaction equation), and 2) the convection/convection-diffusion process (described by a convection/convection-diffusion equation). Under the assumption that $A+B\to AB$ has been finished as soon as species $A$ and $B$ enter a straight microfluidic channel, we can use the solution of the reaction equation as an initial condition for the convection/convection-diffusion equation. 
\vspace{-12pt}
\subsection{The First Approximation Method}
The first method splits (\ref{reaction1}) into a reaction equation and a convection equation by ignoring the diffusion term. The residual concentration of species $A$ is the portion whose concentration is greater than $C_{B_0}^{\text{}}$, and is expressed as 
\begin{align}\label{er6}
C_A(0,t)-C_{B_0}, ~t_1\leq t \leq t_2.
\end{align}
The subsequent transport of species $A$ will be only affected by convection. It has also shown in \cite{hundsdorfer2013numerical} that the convection effect is merely a shift of initial specie profiles in time with velocity $v_\text{eff}$ and without any change of shape, so the outlet concentration of $A$ at the reaction channel can be expressed as
\begin{equation}
	C_{A}^{\text{Appro$_1$}}(x,t)=\begin{cases}C_A(0,t-\frac{x}{v_\text{eff}})-C_{B_0}, &t_1+\frac{x}{v_\text{eff}}\leq t \leq t_2+\frac{x}{v_\text{eff}},\\
	0, &\text{otherwise}.	
	\end{cases}
\end{equation}

\subsection{The Second Approximation Method}    
Different from the first approximation method, the second one takes the diffusion effect into account. 
The convection-diffusion equation with initial condition in (\ref{er6}) and other boundary conditions can be constructed as
\begin{align}
\frac{\partial C_A^{\text{Appro$_2$}}(x,t)}{{\partial}t}=D_{\text{eff}}\frac{\partial^2 C_A^{\text{Appro$_2$}}(x,t)}{{\partial}x^2}-v_{\text{eff}}\frac{\partial C_A^{\text{Appro$_2$}}(x,t)}{{\partial}x}, \label{er8} \\ 
C_A^{\text{Appro$_2$}}(0,t)=C_A(0,t)-C_{B_0}, ~t_1\leq t \leq t_2,\label{er9}\\
C_A^{\text{Appro$_2$}}(x,0)=0, ~x\ge 0, \label{er10}\\
\text{and} \;\;\frac{\partial C_A^{\text{Appro$_2$}}(x,t)}{\partial x}\mid_{x=\infty} =0, ~t\ge 0. \label{er11}
\end{align}
We take the Laplace transform of (\ref{er8}) with respect to $t$ and obtain
\begin{equation}
{D_\text{eff}}\frac{{{\partial ^2}\widetilde{{C_A^{\text{Appro$_2$}}}}(x,s)}}{{\partial {x^2}}} - {v_\text{eff}}\frac{{\partial \widetilde{{C_A^{\text{Appro$_2$}}}}(x,s)}}{{\partial x}} -s \widetilde{{C_A^{\text{Appro$_2$}}}}(x,s)=0. \label{er12}
\end{equation}
The solution of this second order differential equation satisfying 
the Laplace transforms of (\ref{er9}) and (\ref{er11}) is 
\begin{align} \label{er13}
\widetilde{{C_A^{\text{Appro$_2$}}}}(x,s)=l(s)e^{\frac{v_\text{eff}-\sqrt{{v_\text{eff}}^2+4D_{\text{eff}}s}}{2D_{\text{eff}}}x},
\end{align}
where $l(s)$ is a coefficient function and is the Laplace transform of (\ref{er9}), which is
\begin{equation}\label{er14}
\begin{aligned}
l(s)&=\int_{t_1}^{t_2}[C_A(0,t)-C_{B_0}]e^{-st}\mathrm{d}t \\
&={C_{A_0}^{\text{}}}e^{-s \mu}e^{\frac{{(\sigma s)}^2}{2}}[Q({\frac{t_1+\sigma^2 s-\mu}{\sigma}})-Q({\frac{t_2+\sigma^2 s-\mu}{\sigma}})]-\frac{C_{B_0}^{\text{}}}{s}(e^{-st_1}-e^{-st_2}) ,
\end{aligned}
\end{equation}
where $Q(.)$ is the Q-function.

In order to obtain $C_A^{\text{Appro$_2$}}(x,t)$, it is necessary to take the inverse Laplace transform of (\ref{er13}). However, due to the complexity of (\ref{er14}), we cannot derive the closed-form expression $\mathcal{L}^{-1}\left\{C_A^{\text{Appro$_2$}}(x,s) \right\}$. Hence, we employ the Gil-Pelaez theorem \cite{wendel1961non,7511443}. Considering that the Fourier transform of a probability density function (PDF) is its characteristic function, (\ref{er13}) is firstly converted to Fourier transform $\widetilde{{C_A^{\text{Appro$_2$}}}}(x,\omega)$ by substituting $j\omega$ for $s$, and then we regard $\widetilde{{C_A^{\text{Appro$_2$}}}}(x,\omega)$ as the characteristic function of  $\mathcal{L}^{-1}\left\{C_A^{\text{Appro$_2$}}(x,s) \right\}$. The corresponding cumulative distribution function (CDF) can be given in terms of $\widetilde{{C_A^{\text{Appro$_2$}}}}(x,\omega)$ as
\begin{align} \label{er15}
F(t)=\frac{1}{2}-\frac{1}{\pi}\int_{0}^{\infty} \frac{e^{-j\omega t}\overline{\widetilde{{C_A^{\text{Appro$_2$}}}}(x,\omega)}-e^{j\omega t}{\widetilde{{C_A^{\text{Appro$_2$}}}}(x,\omega)}}{2j\omega}\mathrm{d}w.
\end{align}
Taking the derivative of $F(t)$, we derive the inverse Laplace transform and obtain the outlet concentration of speceies $A$ as
\begin{align} \label{er16}
C_A^{\text{Appro$_2$}}(x,t)=\frac{1}{2 \pi}\int_0^\infty [e^{-j\omega t}\overline{\widetilde{{C_A^{\text{Appro$_2$}}}}(x,\omega)}+e^{j\omega t}{\widetilde{{C_A^{\text{Appro$_2$}}}}(x,\omega)}]\mathrm{d}w.
\end{align}

\ifCLASSOPTIONcaptionsoff
  \newpage
\fi


\begin{thebibliography}{10}
		\baselineskip 12pt
		\providecommand{\url}[1]{#1}
		\csname url@samestyle\endcsname
		\providecommand{\newblock}{\relax}
		\providecommand{\bibinfo}[2]{#2}
		\providecommand{\BIBentrySTDinterwordspacing}{\spaceskip=0pt\relax}
		\providecommand{\BIBentryALTinterwordstretchfactor}{4}
		\providecommand{\BIBentryALTinterwordspacing}{\spaceskip=\fontdimen2\font plus
			\BIBentryALTinterwordstretchfactor\fontdimen3\font minus
			\fontdimen4\font\relax}
		\providecommand{\BIBforeignlanguage}[2]{{%
				\expandafter\ifx\csname l@#1\endcsname\relax
				\typeout{** WARNING: IEEEtran.bst: No hyphenation pattern has been}%
				\typeout{** loaded for the language `#1'. Using the pattern for}%
				\typeout{** the default language instead.}%
				\else
				\language=\csname l@#1\endcsname
				\fi
				#2}}
		\providecommand{\BIBdecl}{\relax}
		\BIBdecl
		
		\bibitem{ifa_paradigm}
		I.~F. Akyildiz, F.~Brunetti, and C.~Blazquez, ``Nanonetworks: {A} new
		communication paradigm at molecular level,'' \emph{Comput. Netw.}, vol.~52,
		no.~12, pp. 2260--2279, August 2008.
		
		\bibitem{Akyildiz15}
		I.~F. Akyildiz, M.~Pierobon, S.~Balasubramaniam, and Y.~Koucheryavy, ``The
		internet of bio-nano things,'' \emph{IEEE Commun. Mag.}, vol.~53, no.~3, pp.
		32--40, March 2015.
		
		\bibitem{Farsad16}
		N.~Farsad, H.~B. Yilmaz, A.~Eckford, C.-B. Chae, and W.~Guo, ``A comprehensive
		survey of recent advancements in molecular communication,'' \emph{IEEE
			Commun. Surveys Tuts.}, vol.~18, no.~3, pp. 1887--1919, Third Quarter 2016.
		
		\bibitem{berg1993random}
		H.~C. Berg, \emph{Random walks in biology}.\hskip 1em plus 0.5em minus
		0.4em\relax Princeton, NJ, USA: Princeton Univ. Press, 1993.
		
		\bibitem{bruustheoretical}
		H.~Bruus, \emph{Theoretical microfluidics}.\hskip 1em plus 0.5em minus
		0.4em\relax London, UK: Oxford Univ. Press, 2008.
		
		\bibitem{6712164}
		A.~{Noel}, K.~C. {Cheung}, and R.~{Schober}, ``Improving receiver performance
		of diffusive molecular communication with enzymes,'' \emph{IEEE Trans.
			Nanobiosci.}, vol.~13, no.~1, pp. 31--43, March 2014.
		
		\bibitem{yansha2016stochastic}
		Y.~{Deng}, A.~{Noel}, W.~{Guo}, A.~{Nallanathan}, and M.~{Elkashlan}, ``3d
		stochastic geometry model for large-scale molecular communication systems,''
		in \emph{Proc. IEEE GLOBECOM}, Dec 2016, pp. 1--6.
		
		\bibitem{alberts2013essential}
		B.~Alberts, D.~Bray, K.~Hopkin, A.~D. Johnson, J.~Lewis, M.~Raff, K.~Roberts,
		and P.~Walter, \emph{Essential cell biology}.\hskip 1em plus 0.5em minus
		0.4em\relax New York, NY, USA: Garland Science, 2013.
		
		\bibitem{7511443}
		Y.~{Deng}, A.~{Noel}, M.~{Elkashlan}, A.~{Nallanathan}, and K.~C. {Cheung},
		``Molecular communication with a reversible adsorption receiver,'' in
		\emph{Proc. IEEE ICC}, May 2016, pp. 1--7.
		
		\bibitem{6305481}
		M.~{Pierobon} and I.~F. {Akyildiz}, ``Capacity of a diffusion-based molecular
		communication system with channel memory and molecular noise,'' \emph{IEEE
			Trans. Inf. Theory}, vol.~59, no.~2, pp. 942--954, Feb 2013.
		
		\bibitem{7541454}
		N.~{Farsad}, Y.~{Murin}, A.~{Eckford}, and A.~{Goldsmith}, ``On the capacity of
		diffusion-based molecular timing channels,'' in \emph{Proc. IEEE ISIT}, July
		2016, pp. 1023--1027.
		
		\bibitem{8633972}
		M.~B. {Dissanayake}, Y.~{Deng}, A.~{Nallanathan}, M.~{Elkashlan}, and
		U.~{Mitra}, ``Interference mitigation in large-scale multiuser molecular
		communication,'' \emph{IEEE Trans. Commun.}, pp. 1--1, 2019.
		
		\bibitem{farsad2013tabletop}
		N.~Farsad, W.~Guo, and A.~W. Eckford, ``Tabletop molecular communication: Text
		messages through chemical signals,'' \emph{PLoS One}, vol.~8, no.~12, p.
		e82935, December 2013.
		
		\bibitem{7397863}
		B.~{Koo}, C.~{Lee}, H.~B. {Yilmaz}, N.~{Farsad}, A.~{Eckford}, and C.~{Chae},
		``Molecular mimo: From theory to prototype,'' \emph{IEEE J. Sel. Areas
			Commun.}, vol.~34, no.~3, pp. 600--614, March 2016.
		
		\bibitem{giannoukos2018chemical}
		S.~Giannoukos, D.~T. McGuiness, A.~Marshall, J.~Smith, and S.~Taylor, ``A
		chemical alphabet for macromolecular communications,'' \emph{Anal Chem},
		vol.~90, no.~12, pp. 7739--7746, May 2018.
		
		\bibitem{6630482}
		E.~{De Leo}, L.~{Donvito}, L.~{Galluccio}, A.~{Lombardo}, G.~{Morabito}, and
		L.~M. {Zanoli}, ``Communications and switching in microfluidic systems: Pure
		hydrodynamic control for networking labs-on-a-chip,'' \emph{IEEE Trans.
			Commun.}, vol.~61, no.~11, pp. 4663--4677, November 2013.
		
		\bibitem{6668865}
		B.~{Krishnaswamy}, C.~M. {Austin}, J.~P. {Bardill}, D.~{Russakow}, G.~L.
		{Holst}, B.~K. {Hammer}, C.~R. {Forest}, and R.~{Sivakumar}, ``Time-elapse
		communication: Bacterial communication on a microfluidic chip,'' \emph{IEEE
			Trans. Commun.}, vol.~61, no.~12, pp. 5139--5151, December 2013.
		
		\bibitem{8418677}
		A.~{Marcone}, M.~{Pierobon}, and M.~{Magarini}, ``Parity-check coding based on
		genetic circuits for engineered molecular communication between biological
		cells,'' \emph{IEEE Trans. Commun.}, vol.~66, no.~12, pp. 6221--6236,
		December 2018.
		
		\bibitem{8255057}
		Y.~{Deng}, M.~{Pierobon}, and A.~{Nallanathan}, ``A microfluidic feed forward
		loop pulse generator for molecular communication,'' in \emph{Proc. IEEE
			GLOBECOM}, Dec 2017, pp. 1--7.
		
		\bibitem{weiss2003genetic}
		R.~Weiss, S.~Basu, S.~Hooshangi, A.~Kalmbach, D.~Karig, R.~Mehreja, and
		I.~Netravali, ``Genetic circuit building blocks for cellular computation,
		communications, and signal processing,'' \emph{Nat. Comput.}, vol.~2, no.~1,
		pp. 47--84, Mar. 2003.
		
		\bibitem{cook2009programmability}
		M.~Cook, D.~Soloveichik, E.~Winfree, and J.~Bruck, ``Programmability of
		chemical reaction networks,'' in \emph{Algorithmic Bioprocesses}.\hskip 1em
		plus 0.5em minus 0.4em\relax Springer, 2009, pp. 543--584.
		
		\bibitem{Pehlivanoglu17}
		E.~B. Pehlivanoglu, B.~D. Unluturk, and O.~B. Akan, ``Modulation in molecular
		communications: A look on methodologies,'' in \emph{Modeling, Methodologies
			and Tools for Molecular and Nano-scale Communications}.\hskip 1em plus 0.5em
		minus 0.4em\relax Springer International Publishing, 2017, pp. 79--97.
		
		\bibitem{yansha2017stochastic}
		Y.~{Deng}, A.~{Noel}, W.~{Guo}, A.~{Nallanathan}, and M.~{Elkashlan},
		``Analyzing large-scale multiuser molecular communication via 3-d stochastic
		geometry,'' \emph{IEEE Trans. Mol. Biol. Multi-Scale Commun.}, vol.~3, no.~2,
		pp. 118--133, June 2017.
		
		\bibitem{karlebach2008modelling}
		G.~Karlebach and R.~Shamir, ``Modelling and analysis of gene regulatory
		networks,'' \emph{Nat. Rev. Mol. Cell Biol.}, vol.~9, no.~10, pp. 770--780,
		October 2008.
		
		\bibitem{milo2002network}
		R.~Milo, S.~Shen-Orr, S.~Itzkovitz, N.~Kashtan, D.~Chklovskii, and U.~Alon,
		``Network motifs: simple building blocks of complex networks,''
		\emph{Science}, vol. 298, no. 5594, pp. 824--827, October 2002.
		
		\bibitem{alon2007network}
		U.~Alon, ``Network motifs: theory and experimental approaches,'' \emph{Nat.
			Rev. Genet.}, vol.~8, no.~6, pp. 450--461, June 2007.
		
		\bibitem{alon2006introduction}
		A.~Uri, \emph{An introduction to systems biology: design principles of
			biological circuits}.\hskip 1em plus 0.5em minus 0.4em\relax London, UK:
		Chapman \& Hall, 2006.
		
		\bibitem{mangan2006incoherent}
		S.~Mangan, S.~Itzkovitz, A.~Zaslaver, and U.~Alon, ``The incoherent
		feed-forward loop accelerates the response-time of the gal system of
		escherichia coli,'' \emph{J Mol Biol}, vol. 356, no.~5, pp. 1073--1081, March
		2006.
		
		\bibitem{Kahl13}
		L.~J. Kahl and D.~Endy, ``A survey of enabling technologies in synthetic
		biology,'' \emph{Journal of Biol. Eng.}, vol.~7, no.~1, p.~13, May 2013.
		
		\bibitem{Klinge16modu}
		T.~H. Klinge, ``Modular and robust computation with deterministic chemical
		reaction networks,'' Ph.D. dissertation, 2016.
		
		\bibitem{whitesides2006origins}
		G.~M. Whitesides, ``The origins and the future of microfluidics,''
		\emph{Nature}, vol. 442, no. 7101, pp. 368--373, July 2006.
		
		\bibitem{li2013microfluidic}
		X.~J. Li and Y.~Zhou, \emph{Microfluidic devices for biomedical
			applications}.\hskip 1em plus 0.5em minus 0.4em\relax Sawston, UK: Woodhead,
		2013.
		
		\bibitem{di2009inertial}
		D.~Di~Carlo, ``Inertial microfluidics,'' \emph{Lab. Chip}, vol.~9, no.~21, pp.
		3038--3046, September 2009.
		
		\bibitem{stocker2011introduction}
		T.~Stocker, \emph{Introduction to climate modelling}.\hskip 1em plus 0.5em
		minus 0.4em\relax Berlin, Germany: Springer, 2011.
		
		\bibitem{wicke2018modeling}
		W.~{Wicke}, T.~{Schwering}, A.~{Ahmadzadeh}, V.~{Jamali}, A.~{Noel}, and
		R.~{Schober}, ``Modeling duct flow for molecular communication,'' in
		\emph{Proc. IEEE GLOBECOM}, Dec 2018, pp. 206--212.
		
		\bibitem{bicen2014end}
		A.~O. Bicen and I.~F. Akyildiz, ``End-to-end propagation noise and memory
		analysis for molecular communication over microfluidic channels,'' \emph{IEEE
			Trans. Commun.}, vol.~62, no.~7, pp. 2432--2443, July 2014.
		
		\bibitem{Yansha16}
		Y.~{Deng}, A.~{Noel}, M.~{Elkashlan}, A.~{Nallanathan}, and K.~C. {Cheung},
		``Modeling and simulation of molecular communication systems with a
		reversible adsorption receiver,'' \emph{IEEE Trans. Mol. Biol. Multi-Scale
			Commun.}, vol.~1, no.~4, pp. 347--362, December 2015.
		
		\bibitem{hundsdorfer2013numerical}
		W.~Hundsdorfer and J.~G. Verwer, \emph{Numerical solution of time-dependent
			advection-diffusion-reaction equations}.\hskip 1em plus 0.5em minus
		0.4em\relax New York, NY, USA: Springer, 2013.
		
		\bibitem{wendel1961non}
		J.~Wendel \emph{et~al.}, ``The non-absolute convergence of gil-pelaez'inversion
		integral,'' \emph{The Annals of Mathematical Statistics}, vol.~32, no.~1, pp.
		338--339, 1961.
		
	\end{thebibliography}
\end{document}